\documentclass[a4paper,twocolumn,11pt,accepted=2024-01-17]{quantumarticle}
\pdfoutput=1
\usepackage[utf8]{inputenc}
\usepackage[english]{babel}
\usepackage{soul}
\usepackage[T1]{fontenc}
\usepackage{amsmath}
\usepackage{hyperref}
\usepackage{tcolorbox}
\usepackage{tikz}
\usepackage{lipsum}
\usepackage{physics}
\usepackage{bbold}
\usepackage{amsthm}
\usepackage{float}

\newtheorem{theorem}{Theorem}
\newtheorem{corollary}{Corollary}

\begin{document}

\title{Coherence and contextuality in a Mach-Zehnder \newline interferometer }
\author{Rafael Wagner}
\email{rafael.wagner@inl.int}
\affiliation{International Iberian Nanotechnology Laboratory (INL),
Av. Mestre Jos\'{e} Veiga, 4715-330 Braga, Portugal}
\affiliation{Centro de F\'{i}sica, Universidade do Minho, Braga 4710-057, Portugal}
\orcid{0000-0001-8432-064X}
\author{Anita Camillini}
\email{anita.camillini@inl.int}
\affiliation{International Iberian Nanotechnology Laboratory (INL),
Av. Mestre Jos\'{e} Veiga, 4715-330 Braga, Portugal}
\affiliation{Centro de F\'{i}sica, Universidade do Minho, Braga 4710-057, Portugal}
\orcid{0000-0002-9966-4426}
\author{Ernesto F. Galvão}
\email{ernesto.galvao@inl.int}
\affiliation{International Iberian Nanotechnology Laboratory (INL),
Av. Mestre Jos\'{e} Veiga, 4715-330 Braga, Portugal}
\affiliation{Instituto de F\'{i}sica, Universidade Federal Fluminense,
Av. Gal. Milton Tavares de Souza s/n, Niter\'{o}i, RJ, 24210-340, Brazil}
\orcid{0000-0003-2582-9918}

\maketitle

\begin{abstract}
      We analyse nonclassical resources in interference phenomena using generalized noncontextuality inequalities and basis-independent coherence witnesses. We use recently proposed inequalities that witness both resources within the same framework. We also propose, in view of previous contextual advantage results, a systematic way of applying these tools to characterize advantage provided by coherence and contextuality in quantum information protocols. We instantiate this methodology for the task of quantum interrogation, famously introduced by the paradigmatic bomb-testing interferometric experiment, showing contextual quantum advantage for such a task. 
\end{abstract}

Quantum superposition is the most famous nonclassical feature of quantum theory. It has puzzled generations of physicists with many intriguing interpretations, and it underlies profound discoveries in quantum computation~\cite{shor1999polynomial,parker2000efficient,ahnefeld2022role}, interference of large systems~\cite{nairz2003quantum}, quantum resource theories~\cite{Chitambar19}, quantum complementarity~\cite{bohr1928quantum,wootters1979complementarity,englert1996fringe,cheng2015complementarity,basso2021complete}, and quantum foundations~\cite{elitzur1993quantum,hardy1992existence}. The theory of coherence as a resource in quantum information provided a modern perspective into quantum superposition theory and quantum interference experiments~\cite{baumgratz2014quantifying,Streltsov17}. It gave means for quantifying the amount of coherence present in quantum states inside interferometers, while elegantly characterizing non-classical phenomena{, with better tools} than visibility~\cite{chrysosthemos_quantum_2022,mishra2019decoherence,qureshi2019coherence}, being not only formally grounded in rich resource theoretic results~\cite{biswas_interferometric_2017,paul2017measuring},  and experimentally accessible~\cite{wu2021experimental}, but also profoundly connected with entanglement theory~\cite{streltsov2015measuring,streltsov2016entanglement,qiao2018entanglement}. 

Photonic devices such as beam-splitters can be interpreted as generators of quantum coherence~\cite{masini2021coherence,ares2022beam} and play a major role in modern quantum optics. The abilities of generating, controlling and measuring quantum superposition of single photons in SWAP test and related interferometric schemes are relevant resources in quantum information, particularly useful for learning properties of quantum systems~\cite{ekert2002direct,horodecki2002method,oszmaniec2021measuring}. 

Coherence of quantum states is commonly described with respect to a specific choice of reference basis. Such a choice is normally well-understood in the paradigm of quantum computation, where high control of specific basis states is assumed. Nevertheless, it becomes somewhat arbitrary for states that decohere onto eigenstates of an unknown reference observable. Describing coherence as a basis-dependent quantity is not necessary, and basis-independent coherence can be properly defined for sets of states, in the so-called set-coherence approach{, introduced recently in Ref.}~\cite{designolle2021set}. 

Despite the foundational importance of quantum coherence, quantum superposition and entanglement, those nonclassical phenomena are not sufficient for passing stronger tests of nonclassicality such as violating Bell inequalities, or noncontextuality inequalities~\cite{werner1989quantum,spekkens2007evidence,hardy1999disentangling}. {Mach-Zehnder interferometers (MZIs)} are among the simplest interferometers, and thus a good test-bed for investigating the possibility of demonstrating stronger non-classicality criteria, clarifying when can we witness quantum contextuality in single qubit interference. In Ref.~\cite{catani2021interference} the authors constructed a noncontextual model that reproduces a set of interesting quantum phenomena using a MZI with symmetric beam-splitters and a single phase-shifter with phase $\phi = \pi$. 

In this work we show that we can witness and quantify basis-independent quantum coherence inside Mach-Zehnder interferometers (MZIs) using violations of known coherence-free inequalities, which can be probed using finitely many measurements, as opposed to visibility-based schemes that { rely} on continuously monitoring a quantum system. We also comment on the experimental benefits our approach has with respect to the estimation of basis-dependent coherence monotones~\cite{Streltsov17}.

Our approach allows us to propose \textit{a general way} for demonstrating contextual advantage for quantum information tasks, first envisioned in Ref.~\cite{Lostaglio2020contextualadvantage}. It will be possible to demonstrate advantage arising from coherence/contextuality when a particular task achieves, with a figure of merit, values above thresholds provided by our inequalities.  {Using the inequalities from Refs.~\cite{galvaobroad2020quantumandclassical,wagner2022inequalities} one can then link the figure of merit to such known inequalities to bound classical success rates. }As an instance of this methodology, we consider the case-study of quantum interrogation~\cite{elitzur1993quantum,vaidman1996interaction,kwiat1995interactionfree,kwiat1999high,rudolph2000better}, which exploits a MZI. Using our techniques, we show that quantum contextuality provides a quantifiable advantage for this task. 

{In this work, }we push further the analysis of contextuality of MZIs for any asymmetric beam-splitter and any choice of phase-shifter. {Using the inequality-based framework developed in  Refs.~\cite{galvaobroad2020quantumandclassical,wagner2022inequalities}}, we present experimentally accessible ways to witness and quantify the regime where the superposition created within the interferometer cannot be explained by noncontextual or coherence-free models. Our results show that symmetric beam-splitters in a MZI are the exception, and that virtually any other BS splitting ratio results in phenomena not explainable by a non-contextual model. Furthermore we discuss experimental implementations and their feasibility. 

{In summary, our main contributions are as follows: 
\begin{enumerate}
    \item We apply recently introduced inequalities that witness coherence  to the setting of Mach-Zehnder interferometers, while also remarking that most of our constructions are well-defined for more general interferometric settings. We focus on MZIs, as they may be argued to be the simplest interferometric device.
    \item We provide numerical and analytical results showing that the inequality-based formalism is well suited to experimentally assess quantum coherence inside interferometers.
    \item We present two simple and experimentally accessible settings to probe coherence inside interferometers (the parallel and sequential settings). The parallel setting avoids estimating the visibility and applies also to more general multimode interferometers, while the sequential setting is  suitable to demonstrating the advantage provided by contextuality, as we show. 
    \item We present a novel robust noncontextuality inequality that is capable of probing the ability of MZIs to generate contextual data, that is a generalization of the cycle inequalities found in Ref.~\cite{galvaobroad2020quantumandclassical}. 
    \item We highlight a general strategy for proving quantum advantages in interferometric experiments using the inequalities proposed in Refs.~\cite{galvaobroad2020quantumandclassical,wagner2022inequalities}, and showcase the applicability of this methodology to the task of quantum interrogation.
    \item We apply the known $3$-cycle inequalities from Ref.~\cite{galvaobroad2020quantumandclassical} to quantify contextual advantage in the task just mentioned. This solves an open problem proposed in Ref.~\cite{catani2021interference} of what statistics arising from interference, present in the bomb-testing gedanken experiment, defy explainability by noncontextual models. 
\end{enumerate}
}

This paper is organised in two main sections, the first of which provides the relevant theoretical basis for understanding the main results presented in the second. We start by reviewing the notion of contextuality, and the way in which it is possible to connect contextuality and coherence using the recently proposed event graph approach~\cite{wagner2022inequalities}. We briefly introduce the experimental device we analyze, a MZI, and the information task of quantum interrogation. In the results section afterwards, we first outline how to use coherence witness inequalities to quantify nonclassical behaviour inside the MZI, making a comparison with known results in literature. We then proceed describing how contextuality can be witnessed in those devices, and prove contextual advantage for the task of quantum interrogation.

\section{Background}

Contextuality is a  characteristic of probability distributions. There are three main proposals for defining contextuality. The standard definition, known as Kochen-Specker contextuality~\cite{Budroni21}, is built on the seminal work of Bell, Kochen and Specker~\cite{KochenS67,Bell64,Bell66} on the impossibility of noncontextual hidden-variable models for quantum systems of dimension larger than $2$. The second proposal is known as contextuality-by-default~\cite{dzhafarov2016contextuality,kujala2022contextuality,kujala2019measures}, which extended the analysis beyond physical systems~\cite{cervantes2018snow}. And finally, the generalized contextuality description, first proposed by Spekkens in Ref.~\cite{spekkens2005contextuality}. In this work, we take contextuality to be defined in the Spekkens sense only, a choice that will become clear later.

\subsection{Generalized Contextuality}\label{subsec: generalized contextuality}
We can theoretically and experimentally analyse generalized contextuality using two approaches, (i) via the characterization of operational-probabilistic prepare-and-measure contextuality scenarios with finite operational equivalences~\cite{schmid2018all,chaturvedi2021characterising,tavakoli2021bounding,schmid2018discrimination,Lostaglio2020contextualadvantage,kunjwal_anomalous_2019}, and (ii) via the characterization of a prepare-and-measure fragment of general probabilistic theories (GPTs)~\cite{schmid2021characterization,shahandeh2021contextuality,selby2021accessible,selby2022open}. In Ref.~\cite{spekkens2005contextuality}, Spekkens defined generalized noncontextuality as a property of ontological models~\cite{leifer2014isthequantum,liang2011specker} explaining an operational-probabilistic theory (OPT). We will only treat the former well-established approach and leave  connections and applications of the results for the latter to future work. 

The finite scenarios approach is a technique for selecting finite sets of accessible processes and equivalences in the OPT for probing generalized contextuality as originally proposed. We hereby review the necessary tools for understanding the operational content of prepare-and-measure contextuality scenarios, and how it is possible to probe contextuality in this setting. We follow Refs~\cite{schmid2018all,chaturvedi2021characterising,tavakoli2021bounding,schmid2018discrimination,Lostaglio2020contextualadvantage,lostaglio_certifying_2020,kunjwal2019beyondcabello,kunjwal_anomalous_2019}, mostly inspired by the informal description of operational theories from Ref.~\cite{spekkens2005contextuality}; see Ref.~\cite{schmid2020structure} for the distinction between OPTs and GPTs we will assume.

An OPT prescribes lists of operations, subdivided into preparations $\mathcal{P}$, transformations $\mathcal{T}$, and measurements $\mathcal{M}$ with associated effects{, i.e. ordered pairs of outcomes and their corresponding measurements,} $\{[k \vert M]\}_k${, with $M \in \mathcal{M}$ and $k$ some outcome of $M$}. Elements in these lists correspond to laboratory instructions to be done during an experiment, e.g. an element $P_0 \in \mathcal{P}$ can represent the instruction ``prepare a single photon in mode \textit{a}'' and another element $P_{0^\perp} \in \mathcal{P}$ can represent the instruction  ``prepare a single photon in mode \textit{b}'' (see Fig.~\ref{fig:MZI}). In theory, all these lists can have infinite elements. An OPT is also assumed to be convex~\cite{schmid2020structure} and to have consistent composition operations.  Using this framework we must also postulate  a probability rule between preparations followed by measurements $p(k\vert M,T,P)$. Indistinguishable operations of an OPT with respect to the rule $p$ are said to be equivalent. 

Formally, any two $P_1,P_2 \in \mathcal{P}$ are operationally equivalent if $p(k\vert M,T,P_1) = p(k \vert M,T,P_2)$ for all conceivable effects $\{[k \vert M]\}_k$, $M \in \mathcal{M}$ and transformations $T \in \mathcal{T}$ in the OPT, and we write $P_1 \simeq P_2$. Similarly for measurement effects{, any two effects $[k_1|M_1], [k_2|M_2]$ are said to be operationally equivalent if $p(k_1|M_1,T,P) = p(k_2|M_2,T,P)$ for all conceivable preparations $P \in \mathcal{P}$ and transformations $T \in \mathcal{T}$, in which case we write $[k_1|M_1] \simeq [k_2|M_2]$}. In words, equivalent processes are indistinguishable through the lenses of the operational theory. We will not consider operational equivalences between transformation procedures, since we will assume that operations of the form $T(P)$ correspond to valid new preparation procedures within the OPT, similar to Refs.~\cite{Lostaglio2020contextualadvantage,baldijao_emergence_2021}.

When quantum theory is viewed as an operational theory, quantum states $\rho \in \mathcal{D}(\mathcal{H})$, over some space $\mathcal{H}$, label equivalence classes of operationally equivalent preparation procedures, and POVM elements label equivalence classes of measurement effects. Hence, operational equivalences $P_1\simeq P_2$ are represented in quantum theory as equality between corresponding states $\rho_1 = \rho_2$, and $[k_1\vert M_1]\simeq [k_2\vert M_2]$ corresponds to equality of POVM elements $E_{k_1}^{M_1} = E_{k_2}^{M_2}$.

The operational theory does not attempt to explain the phenomenology of a given experiment, since it merely provides the set of rules for obtaining the statistical results thereof. Explanations for the success of an operational theory are given by the ontological models framework~\cite{liang2011specker,leifer2014isthequantum}. Such models are constructed by postulating the existence of a set $\Lambda$ of physical variables $\lambda \in \Lambda$, that is relevant for the explanations of the phenomena that constitute the elements of reality~\cite{EinsteinPR35} of the theory. Let us now introduce the ontological counterparts of the operational elements $P, T$ and $k \vert M $, i.e. functions acting over $\Lambda$. Preparations $P$ are associated with the probability $\mu_P$ that a given $\lambda$ is prepared, $\int_\Lambda \mu_P(\lambda)d\lambda = 1, \mu_P(\lambda) \geq 0$. Transformations $T$ label statistical changes from a given $\lambda$ into a different $\lambda'$, and are represented as stochastic transition matrices over $\Lambda$. Finally, measurement effects correspond to functions $\xi$ which describe, for each given $\lambda \in \Lambda$, probability distributions over all the outcomes of the measurement $M$, {$\sum_{k }\xi(k\vert M,\lambda)=1$, $\xi(k\vert M,\lambda) \geq 0$.} This causal description within the variables of $\Lambda$ explains the operational-probabilistic theory when it agrees with experiments, i.e. it satisfies
\begin{equation}
    {p(k \vert M,T,P) = \int \xi(k \vert M,\lambda)\Gamma_T(\lambda \vert \lambda')\mu_P(\lambda') d\lambda d\lambda'.}
\end{equation}

{ The probability $\mu_P$ plays a central role in quantum foundations. When $P$ corresponds to the preparation of some quantum state, say $\vert \psi \rangle$, the associated distributions $\mu_{\psi}$ are referred to as the \textit{epistemic states} of the ontological model. In such a case, epistemic states carry only some knowledge about reality, itself described by the ontic states $\lambda$. The theories for which quantum states are interpreted as epistemic states are generically known as \textit{epistemic interpretations of quantum mechanics}. For an overview of such a discussion we refer to Ref.~\cite{leifer2014isthequantum}. Some celebrated results in quantum foundations are presented in this language, such as Spekkens' noncontextual toy theory~\cite{spekkens2007evidence}, and the Pusey-Barrett-Rudolph (PBR) no-go theorem~\cite{pusey2012onthereality}. }

So far, ontological models as described can explain \textit{any} statistics arising from an OPT, be it classical, quantum, or even post-quantum.  Constraining ontological explanations to satisfy specific assumptions is then necessary to possibly create a gap between classical and quantum explanations. To do so, we constrain operationally equivalent procedures $P_1 \simeq P_2$. Distinctions in the labels, $1$ or $2$, of such equivalent procedures label the elements in the same equivalence class $P_1 , P_2 \in [P_1] = [P_2]$. Each choice of label corresponds to the context in which the procedure was made. For instance, the procedure $P_{(0\text{ or }1)}$ is described as ``flip a fair coin to decide if preparing a single photon in mode \textit{a} or in mode \textit{b}'', and the procedure $P_{(+\text{ or }-)}$ is described as ``flip a fair coin to decide if preparing a single photon in mode \textit{a} that later passes through a symmetric BS, or in mode \textit{b} that later passes through a symmetric BS'', see Fig.~\ref{fig:MZI} for this example. In quantum theory the first procedure prepares the equal mixture of states $\frac{1}{2}\vert 0 \rangle \langle 0 \vert + \frac{1}{2}\vert 1 \rangle \langle 1 \vert = \frac{1}{2}\mathbb{1}$, and the second prepares the equal mixture of states $\frac{1}{2}\vert + \rangle \langle + \vert + \frac{1}{2}\vert - \rangle \langle - \vert = \frac{1}{2}\mathbb{1}$. These two procedures are operationally equivalent, since for every POVM element we will have $p(k\vert M, P_{(0\text{ or }1)}) = \text{Tr}(E_k^M \frac{1}{2}\mathbb{1}) = p(k\vert M, P_{(+\text{ or }-)})$. The operational procedures were nevertheless performed by different elements in the OPT. 

It can be argued that this difference in labels is not acceptable from a classical perspective~\cite{spekkens2005contextuality,spekkens2019ontological}, and disregarding labels of operationally equivalent processes is the relevant constraint over ontological models to be deemed classical. That is precisely what defines the notion of noncontextuality. Formally, we say that an ontological model is preparation noncontextual if
\begin{equation}\label{eq: prep noncontextuality}
    P_1 \simeq P_2 \implies \mu_{P_1}(\lambda) = \mu_{P_2}(\lambda),\quad \forall \lambda \in \Lambda.
\end{equation}
and measurement noncontextual if {
\begin{align}\label{eq: measurement noncontextuality}
    &[k_1\vert M_1] \simeq [k_2 \vert M_2]\nonumber\\
    &\implies \xi(k_1\vert M_1,\lambda)=\xi(k_2\vert M_2,\lambda),\quad \forall \lambda \in \Lambda.
\end{align}
}
The noncontextual constraint just described is consistent with classical mechanics~\cite[Supp. Material A]{lostaglio_certifying_2020}, and with the emergence of classical objectivity~\cite{baldijao_emergence_2021}. Moreover, it is grounded by philosophical desiderata~\cite{spekkens2019ontological}, is experimentally robust~\cite{mazurek2016experimental,mazurek2021experimentally}, reduces to Kochen-Specker noncontextuality under specific conditions, see~\cite[Section 1.3.2]{kunjwal2016contextuality} or~\cite{leifer2013maximally}, and generalises to the GPT framework~\cite{schmid2021characterization,shahandeh2021contextuality}. When there exists no noncontextual ontological model explaining the given OPT, we say that the theory is contextual. 

When viewed as an operational theory, quantum theory is contextual for preparations, transformations and measurements~\cite{spekkens2005contextuality,banik2014ontological,lillystone2019contextuality,kunjwal2019beyondcabello}. OPTs have possibly infinitely many procedures and infinitely many operational equivalences; each prepare-and-measure contextuality scenario, for which we write $\mathbb{B}$, corresponds to finite instances of the underlying OPT being experimentally probed. A  scenario has finite sets of preparations $\{P_x\}_{x \in \underline{X}}$, $\underline{X} := \{1,2,\dots,X\}$, measurements $\{M_y\}_{y \in \underline{Y}}$, and outcomes $\{k\}_{k \in \underline{K}}$, as well as \textit{finite} sets of operational equivalences for the preparations $\mathcal{OE}_P$ and for the measurement effects $\mathcal{OE}_M$. We will represent these scenarios with tuples $\mathbb{B} = (X,Y,K,\mathcal{OE}_P,\mathcal{OE}_M)$.  

Such scenarios have had their noncontextual characterization fully described in terms of complete sets of inequalities~\cite{schmid2018all}, were applied to proofs of contextual advantage~\cite{schmid2018discrimination,Lostaglio2020contextualadvantage,lostaglio_certifying_2020}, investigated using resource theoretic approaches~\cite{Duarte18,wagner2021using}, and fully analysed in terms of the quantum set of correlations using semi-definite programming techniques~\cite{chaturvedi2021characterising,tavakoli2021bounding} inspired by similar approaches to quantum non-local correlations~\cite{navascues2007bounding}. Since $\mathbb{B}$ is characterized by finite sets, the rule $p$ applied to the prepare-and-measure elements of $\mathbb{B}$ defines tuples $\mathbb{R}^{XYK} \ni \mathbf{p} = (p(k \vert M_y,P_x))_{k \in \underline{K}, x \in \underline{X}, y \in \underline{Y}}$, usually called \textit{behaviours}. The set of noncontextual behaviours $\mathcal{NC}$, for a given scenario $\mathbb{B}$, corresponds to the set of all points $\mathbf{p}$ that can be explained by some noncontextual ontological model. { This set forms a convex polytope~\cite{schmid2018all}~\footnote{{We review some basic aspects of the theory of convex polytopes in Sec.~\ref{subsec: event graphs}.}}.} The set of all quantum behaviours $\mathcal{Q}$, for a given scenario $\mathbb{B}$, corresponds to all points $\mathbf{p}$ that can be represented using quantum theory as an operational theory, i.e. there exist sets of states $\{\rho_x\}_{x \in \underline{X}}$ and effects $\{\{E_k^y\}_{k \in \underline{K}}\}_{y \in \underline{Y}}$ satisfying the operational equivalences of $\mathbb{B}$ such that $p(k \vert M_y,P_x) = \text{Tr}(\rho_x E_k^y)$. 

For connecting such operational scenarios with coherence and contextuality inside the interferometers, we will show a parallel between those and the inequalities that were first proposed to rule out coherence-free models for quantum states. In the following, we introduce the inequalities that witness basis-independent coherence for a given set of quantum states.

\subsection{Event graph formalism}\label{subsec: event graphs}
\begin{figure}[bt]
    \centering
    \includegraphics[width=0.2 \textwidth]{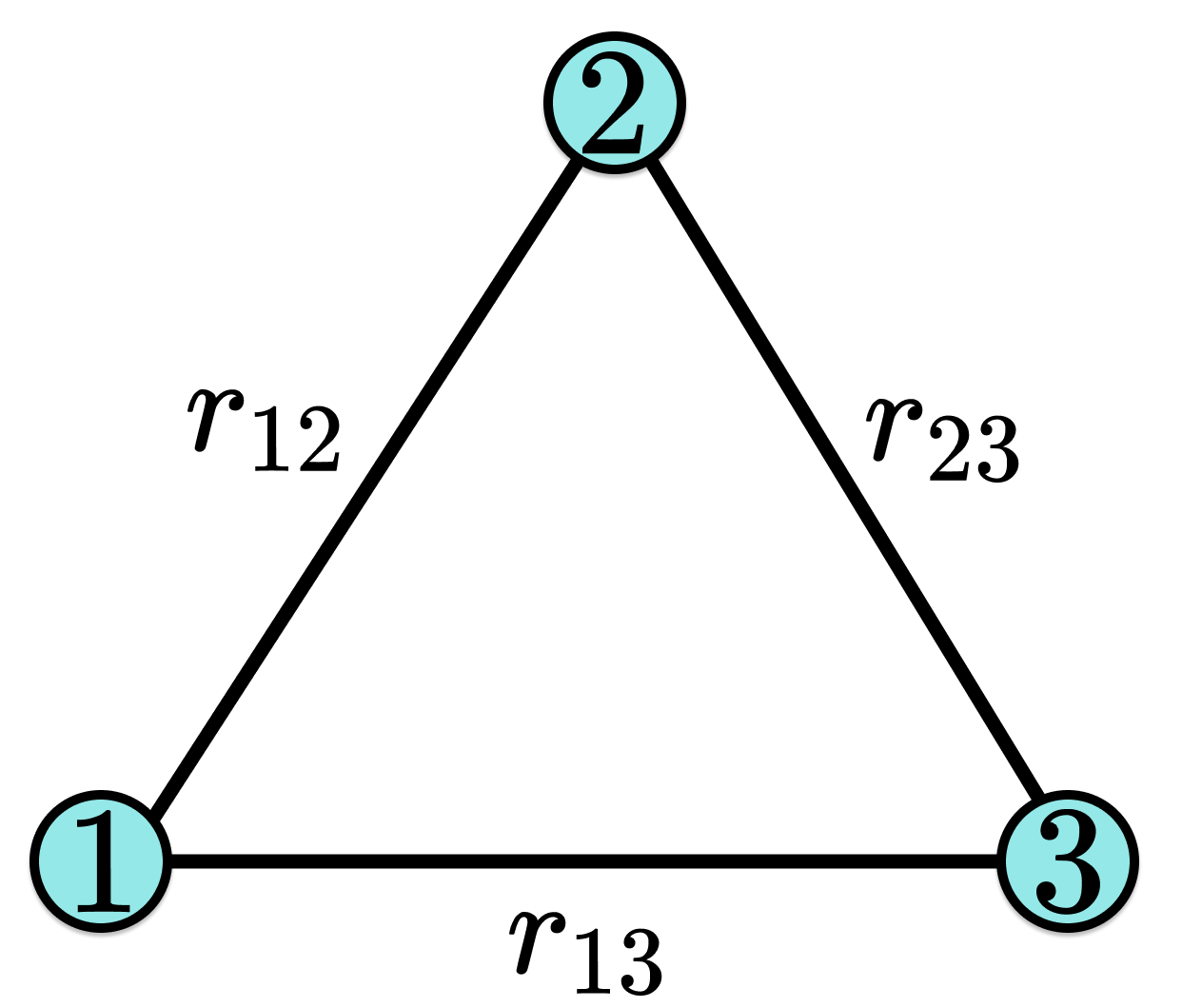}
    \caption{\textbf{Simplest non-trivial event graph.} The $C_3$ cycle graph is an event graph with possible edge weights $r_{12},r_{13},r_{23} \in [0,1]$. Classical weights are convex combinations of deterministic assignments satisfying transitivity of equality.}
    \label{fig: simplest event graph}
 \end{figure}
{ In this section, we review the} event graph approach introduced in Ref.~\cite{wagner2022inequalities} for witnessing basis-independent coherence, contextuality, and nonlocality, which builds on  results by Boole~\cite{Boole1854}, $n$-cycle inequalities~\cite{araujo2013all},  two-state overlap inequalities~\cite{galvaobroad2020quantumandclassical}, and well-established graph-approaches to contextuality~\cite{amaral2018graph,cabello2014graph}. {In this formalism we use graph theory, and the theory of convex polytopes, to obtain inequalities whose violation witnesses coherence based solely on the statistics of two-state overlaps. Let us first introduce the mathematical tools needed for the presentation of the event graph approach: graph theory and the theory of convex polytopes. 

A graph $G$ is an ordered pair $(V,E)$ of two sets. Elements of $V$ are called \textit{vertices} or \textit{nodes} of the graph, while  $E$ is a set of pairs $e \equiv \{v,w\}$, with $v,w \in V$, and which  $e \in E$ elements are called \textit{edges} of the graph. A graph is said to be finite if $|V|<\infty$ and complete if $E = \{\{v,w\}\,|\, \forall v,w \in V\}$, i.e., $E$ has every possible pair of elements from $V$. It is possible to depict graphs as we do in Fig.~\ref{fig: simplest event graph}, where each element from $V$ is depicted as a node (circle) while each element $\{v,w\}$ of $E$ is depicted as a line drawn between nodes $v$ and $w$. A graph is said to be fully connected if there exists a way to move along any two vertices of $V$ using the edges of the graph. Graphs described as above are sometimes called \textit{simple}, as the edges are undirected (i.e., $\{v,w\} = \{w,v\}$), nodes have no loops (i.e., $\{v,v\} \notin E$), and two nodes can pertain to at most one edge (i.e. $\forall e_1,e_2 \in E$ such that  $v,w \in e_1 \cap e_2$ with $v \neq w$, then $e_1 = e_2$). We say that $G' = (V',E')$ is a subgraph of $G = (V,E)$ if $V' \subseteq V$ and $E' \subseteq E$, and denote the set of all possible subgraphs of $G$ as $\text{sub}(G)$. Finally, we say that two graphs $G_1,G_2$ are isomorphic, and write $(V_1,E_1) = G_1 \simeq G_2 = (V_2,E_2)$ if there exists a bijective function $f: V_1 \to V_2$ such that $\forall v_1,u_1 \in V_1$, $\{v_1,u_1\} \in E_1 \iff \{f(v_1),f(u_1)\} \in E_2$.  We will refer to fully connected simple graphs as \textit{event graphs}.}

{ The most important example of event graphs in our work will be the \textit{$n$-cycle graphs} $C_n$, where $n$ represents the number of vertices of the graph, with $V = \{1,2,\dots,n\}$. Edges in these graphs are defined by $E := \{i,i+1\}_{i \in V}$, where summation is defined module $n$. The simplest such case is depicted in Fig.~\ref{fig: simplest event graph}, showing the $3$-cycle graph. }

{ To each event graph it is possible to associate a specific region of the real space $\mathbb{R}^N$, with $N=|E|$, known as a convex polytope. Given some finite set $S\subset \mathbb{R}^N$ of tuples, for some $N \in \mathbb{N}$, \textit{convex polytopes} are defined as the convex hull $\text{ConvHull}(S)$ of this set, $$\left\{\sum_{s \in S}\alpha_s s \in \mathbb{R}^N \, : \, \sum_{s \in S}\alpha_s = 1, \\ 0 \leq \alpha_s \leq 1, \forall s \in S\right\},$$ which is the so-called V-representation of the polytope. Such polytopes are the higher-dimensional generalization of polygons. Equivalently, the same finite convex subset of $\mathbb{R}^N$ can be described in terms of a finite collection of  closed half-spaces $H := \{r \in \mathbb{R}^N: h(r) \leq b, b \in \mathbb{R}\}$, where $h:\mathbb{R}^N \to \mathbb{R}$ is a linear functional. Each $h(r) \leq b$ defines the so-called facet inequalities for the convex polytope. This collection of half-spaces is referred to as the H-representation. In our work, all the inequalities that we will use to quantify coherence inside interferometers correspond to facet-defining inequalities of a certain convex polytope that we now proceed to define. }

To any event graph $G$ we assign weights $r_e \in [0,1]$ to the edges of $G$. The full set of possible edge {weights} forms the  convex polytope $[0,1]^{|E(G)|}$. {The event graph framework bounds coherence and contextuality depending on the possible \textit{realizations} of these weights, that arise under a specific encoding of quantum information into the event graph. We define a \textit{quantum realization} for a given event graph $G$ with weights $r=(r_e)_e$ whenever we can associate the vertices of the graph $V$ to density matrices $V \ni v \leftrightarrow \rho_v \in \mathcal{D}(\mathcal{H})$ over \textit{some} Hilbert space $\mathcal{H}$, such that $r_{e} = \text{Tr}(\rho_i\rho_j)$,  $\forall e=\{i,j\}$. } 

{At this point, quantum realizations simply constitute a mathematical construction, lacking an operational interpretation. They may not exist, as for instance for the edge-weights $(1,0,1)$ in the $3$-cycle graph $C_3$ of Fig.~\ref{fig: simplest event graph}. In general, there can be an infinite number of different realizations for the same set of edge-weights. For example, the triple of weigths $(1,1,1)$, again defined with respect to the graph $C_3$, can be realized making all  nodes equal to $\vert \psi \rangle \langle \psi \vert $, for any $\vert \psi \rangle \in \mathcal{H}$. We will consistently use the same terminology for quantum realizations in the coming sections~\ref{subsec: parallel setting}, ~\ref{subsec: sequential setting} and ~\ref{subsec: Generic advantages}.

The way in which one can experimentally \textit{infer} quantum realizations may differ. For consistency, in our work all the two-state overlaps will be operationally interpreted as arising from prepare-and-measure statistics, with 
\begin{equation*}
p(0|M_\sigma,P_\rho) = \text{Tr}(\rho \sigma)
\end{equation*} obtained from preparing the state $\rho$, denoted by $P_\rho$, and measuring $M_\sigma := \{\sigma, \mathbb{1}-\sigma\}$ with outcome $\sigma$, denoted as $[0|M_\sigma]$. It is still worth pointing out that the interpretation of coherence witnesses from Sec.~\ref{subsec: parallel setting} remains valid if overlaps are estimated using different techniques, e.g., Hong-Ou-Mandel test statistics in photonics ~\cite{giordani2021witnessesofcoherence}.}

{We will quantify coherence from the basis-independent perspective~\cite{designolle2021set,galvaobroad2020quantumandclassical}, that considers coherence as a property of a set of states, defined as follows: Let $\{\rho_i\}_i$ be a finite set of quantum states defined over the same finite Hilbert space $\mathcal{H}$. We say that this set is \textit{coherence-free}, or \textit{set-incoherent}, if all states in the set are mutually commuting operators. Otherwise, we say that such an ensemble of states is \textit{set-coherent}. Equivalently, a set of states is coherence-free if there exists some basis $\Omega$ with respect to which all states can be represented as diagonal density matrices.
}

{Let us consider the quantum realization of a given event graph for which the set of states is coherence-free. In this case, one can still use the same prepare-and-measure protocol to estimate the two-state overlaps, but the existence of a reference observable allows also for a different measurement protocol for estimating the overlaps. }If all the states are incoherent with respect to some basis $\Omega$, then $\rho_i = \sum_{\omega \in \Omega} p(\omega \vert i)\vert \omega \rangle \langle \omega \vert $, which implies \begin{equation}\label{eq: classical overlap}
    r_{ij} = \sum_{\omega \in \Omega} p(\omega \vert i)p(\omega \vert j).
\end{equation} Any such $r_{ij}$ could be interpreted as the probability that, when the two states are measured with the basis $\Omega$, they return equal outcomes. This means that we can see incoherent sets of states implying two-state overlaps, such that $r_{ij} = p(v(\Omega)_{\rho_i} = v(\Omega)_{\rho_j})$, where $v(\Omega)_{\rho_i}$ is the value returned by $\Omega$ measured over $\rho_i$. 

{ The set of all possible overlap tuples $r = (r_e)_{e \in E}$ associated with a given event graph $G=(V,E)$ that are realized by set-incoherent quantum states in the nodes of the graph form an associated convex polytope $C_G$. We now describe the construction of this polytope, for each event graph $G$.}

{
Let $\mathbb{H}$ denote the vertices of the polytope $[0,1]^{|E|}$, defined as all possible $0/1$ edge-weight assignments to an event graph $G$. We call such assignments for the edge-weights $r = (r_e)_e$ deterministic. If there exists at least one edge $e' \in E$ such that $r_e \in (0,1)$ we say that the weight $r$ is probabilistic. Let \begin{align*}
\mathcal{L}(G) &:= \{ E' \subseteq E : \exists n \in \mathbb{N}_{\geq 3},\\&\exists G' \in \text{sub}(G) \text{ such that } G' \simeq C_n\}
\end{align*}denote the set of all $E'$ from subgraphs $G'=(V',E') \in \text{sub}(G)$  isomorphic to some cycle graph $C_n$, for some $n\geq 3$. Denoting $r\vert_{\ell}$ the restriction of the edge-weights $r: E \to [0,1]$ onto some subset of edges $\ell \in \mathcal{L}(G)$ and define the vertex set $\mathbb{V}_G$ as
\begin{equation}
    \mathbb{V}_G := \{r \in \mathbb{H}: |\text{root}(r|_{\ell})| \neq 1, \forall \ell \in \mathcal{L}(G)\}
\end{equation} 
where $\text{root}(r) := \{e \in E: r_e = 0\}$. In words, $\mathbb{V}_G$ has all deterministic edge-assignments satisfying that for any cycle $\ell$ of $G$ there is no single $0$ edge-assignment. For instance, $\mathbb{V}_{C_3}$, with $C_3$ the graph shown in Fig.~\ref{fig: simplest event graph},  is given by,
\begin{equation*}
    \mathbb{V}_{C_3}:=\left\{ \left(\begin{matrix}1\\1\\1\end{matrix}\right), \left(\begin{matrix}0\\0\\0\end{matrix}\right), \left(\begin{matrix}1\\0\\0\end{matrix}\right),\left(\begin{matrix}0\\1\\0\end{matrix}\right), \left(\begin{matrix}0\\0\\1\end{matrix}\right) \right\}
\end{equation*}
In the event graph framework, the convex polytope $C_G$ is defined by 
\begin{equation}
    C_G := \text{ConvHull}(\mathbb{V}_G).
\end{equation}
Since $\mathbb{V}_G \subseteq \mathbb{H}$ for all event graphs $G$, the polytope $C_G$, for any event graph $G = (V,E)$, is a subpolytope of the hypercube $[0,1]^{|E|}$. This is simply because while the hypercube is the convex hull of all possible $2^{|E|}$ $0/1$ edge-assignments for $G$, the polytope $C_G$ will be defined as the convex hull of only a subset of possible $0/1$ edge-assignments. The possible ones will be those that satisfy the hypothesis of quantum realization by some set of incoherent states, with respect to some Hilbert space dimension, described by Eq.~\eqref{eq: classical overlap}.
}

{
While the description just given for the convex polytopes $C_G$ is formal, it does not provide an intuitive understanding of why any quantum realization with incoherent states should return overlap tuples $r \in C_G$. For that, notice that any incoherent quantum realization for the overlaps of the form given by Eq.~\eqref{eq: classical overlap} satisfies the property of \textit{transitivity of equality}, which is a logical consistency condition necessary for classical probability distributions. This condition provides intuition on why we consider loops $\ell \in \mathcal{L}(G)$ to construct $C_G$. 

As we have shown, for an incoherent quantum realization two-state overlaps can also be interpreted as $r_{ij} = p(v(\Omega)_{\rho_i} = v(\Omega)_{\rho_j})$, which is the probability that upon measuring $\Omega$ independently on each incoherent state in the nodes $i$ and $j$ the outcomes were equal. In case we have an incoherent quantum realization satisfying $r_{ij}=1$, we must conclude that the two states are indistinguishable under $\Omega$. Let $\ell \in \mathcal{L}(G)$ for any event graph $G$, if $|\text{root}(r|_{\ell})| = 1$ and $r$ is deterministic, there exists a single edge $\{i,j\} \in \ell$ such that $r_{ij}=0$, while all others are equal to 1. Under the hypothesis that $r$ is realized by incoherent states, $|\text{root}(r|_{\ell})| = 1$ implies a logical contradiction, as we must conclude that all states are indistinguishable, while due to $r_{ij}=0$ the states $\rho_i$ and $\rho_j$ must be different.  
} For instance, with reference to Fig.~\ref{fig: simplest event graph}, an assignment $(1\,1\,0)$ is impossible for incoherent states, as it leads to a logical contradiction{, since it would imply that states $\rho_1,\rho_2$ are indistinguishable, and the same for  $\rho_1,\rho_3$. By transitivity, we must conclude that $\rho_1,\rho_3$ are also indistinguishable, but this is in contradiction with the fact that $r_{23}=0$}.

Any set of quantum states realizing edge-weights $r$ outside $C_G$ is then {set-coherent, i.e. the set of states is coherent} \textit{with respect to any possible basis}. This formalism for witnessing coherence was first proposed in Ref.~\cite{galvaobroad2020quantumandclassical}, and it was experimentally verified in Ref.~\cite{giordani2021witnessesofcoherence}.

By construction, { the set of allowed points $r = (r_e)_e$ for any given event graph realized by incoherent states corresponds to the proper subset $C_G$ of $[0,1]^{|E|}$, characterized in terms of overlap inequalities  (from which we select some specific ones to study and that we present later on in this section, namely, in Eqs.~\eqref{eq: c3 ineq 1}-\eqref{eq: c3 ineq 3}, Eq.~\eqref{eq: n-cycle inequalities} or Eq.~\eqref{eq: inequality (d)}). In this manner, violations of the aforementioned inequalities constitute witnesses of basis-independent coherence for any set of states realizing the corresponding edge-weights. To conclude, it is also worth mentioning that} there is no dimension constraint as well; coherence for any set of high-dimensional states can be probed using these inequalities.

Furthermore,  Ref.~\cite{wagner2022inequalities} showed that the inequalities associated to cycle graphs from the event graphs are preparation noncontextuality inequalities for specific prepare-and-measure scenarios. Of particular importance for our purposes is the $C_3$ graph shown in Fig.~\ref{fig: simplest event graph}. Its facets correspond to the { $3$-cycle }inequalities
\begin{align}
    &+r_{12}+r_{13}-r_{23} \leq 1,\label{eq: c3 ineq 1}\\
    &+r_{12}-r_{13}+r_{23} \leq 1,\label{eq: c3 ineq 2}\\
    &-r_{12}+r_{13}+r_{23} \leq 1,\label{eq: c3 ineq 3}
\end{align}
which are noncontextuality inequalities under considerations to be clarified in the results section. { We refer to the $n$-cycle inequalities as the facet-defining inequalities of the convex polytope $C_{C_n}$ obtained from the event graphs $C_n = (V(C_n), E(C_n))$, for any $n\geq 3$, defined for any fixed $e' \in E(C_n)$ as
\begin{equation}\label{eq: n-cycle inequalities}
    -r_{e'} + \sum_{\begin{array}{c}e \in E(C_n)\\ e\neq e' \end{array}}r_e \leq n-2.
\end{equation}
These two-state overlap inequalities were first introduced in Ref.~\cite{galvaobroad2020quantumandclassical}, but their derivation relate to other well-known $n$-cycle inequalities in contextuality theory~\cite{araujo2013all}.}

The results of Ref.~\cite{wagner2022inequalities} are even stronger: \textit{every} inequality of $C_G$ for \textit{every} choice $G$ can be viewed as a Kochen-Specker noncontextuality inequality for some scenario. We choose to only focus on Spekkens contextuality because we want to study \textit{single qubit interference}, for which the Kochen and Specker approach has a noncontextual explanation~\cite{KochenS67}. 

Concretely, we focus on the interpretations of $C_3$ as basis-independent coherence witness and \textit{generalized noncontextual} bounds because, (i) it is the simplest non-trivial inequality that has a clear link to generalized contextuality, (ii) a single notion of contextuality aids consistency and readability, (iii) it allows to analyse single-qubit interference, for which the Kochen-Specker theorem does not apply, (iv) { it allows } to probe nonclassicality focusing on states/preparations, and (v) {  also allows} quantum states to be mixed, due to experimental imperfections. This choice marks a conceptual separation between inequalities that we will use to probe contextuality and inequalities we will use to probe coherence. We will focus on $C_3$ inequalities interpreted as generalized noncontextuality bounds and general $C_G$ inequalities as basis-independent coherence witnesses. 

There are several inequalities that might be used to probe coherence instead of cycle inequalities. Interesting inequalities should use a small number of states and overlaps, have high violations, and be achievable with qubits, as this allows for violations using the two arms of a MZI. For three states, the only inequalities are the triangle inequalities in Eqs~\eqref{eq: c3 ineq 1}, \eqref{eq: c3 ineq 2} and \eqref{eq: c3 ineq 3},
with violations up to a maximum of $0.25$. Since for four states the only non-trivial non-cycle inequality does not have qubit violations~\cite{wagner2022inequalities}, for probing larger violations using few states one needs to consider five-state inequalities. For instance, if $G = K_5$ is the complete graph of five nodes, the following inequality{, first derived in Ref.~\cite{wagner2022inequalities},} satisfies all the desired criteria { outlined above}, as we shall show later { in Sec.~\ref{subsec: parallel setting}},
\begin{align}
    &r_{12}+r_{15}+r_{23}+r_{34}+r_{45}\nonumber\\
    &\hspace{3em}-r_{13}-r_{14}-r_{24}-r_{25}-r_{35}\leq 2. \label{eq: inequality (d)}
\end{align}

We can quantify the amount of basis-independent coherence by using the degree of violation of those inequalities~\cite{giordani2020experimental}. Inequality violations have long been used for quantifying nonclassicality, and can be made rigorous lower bounds for distance-based resource theoretic monotones, similar to what was done in Ref.~\cite{brito2018quantifying} (a brief comment on this issue can be found in Appendix A).

\subsection{Mach-Zehnder Interferometer}\label{subsec: MZI}

A Mach-Zehnder interferometer (MZI) is a particularly simple device capable of demonstrating the wave-like behavior of photons~\cite{loudon2000quantum,zetie2000does}. This device is present in many discussions in quantum foundations~\cite{hardy1992existence,elitzur1993quantum,catani2021interference}. In its standard configuration, a MZI is made of two beam-splitters (BSs) and two mirrors, with optical paths of equal length. Quantum information is encoded on the photons' path, and single photon interference is captured by variations of a phase-shifter (PS) tuned inside one of the arms. 

\begin{figure}[tb]
    \centering
    \includegraphics[width=\columnwidth]{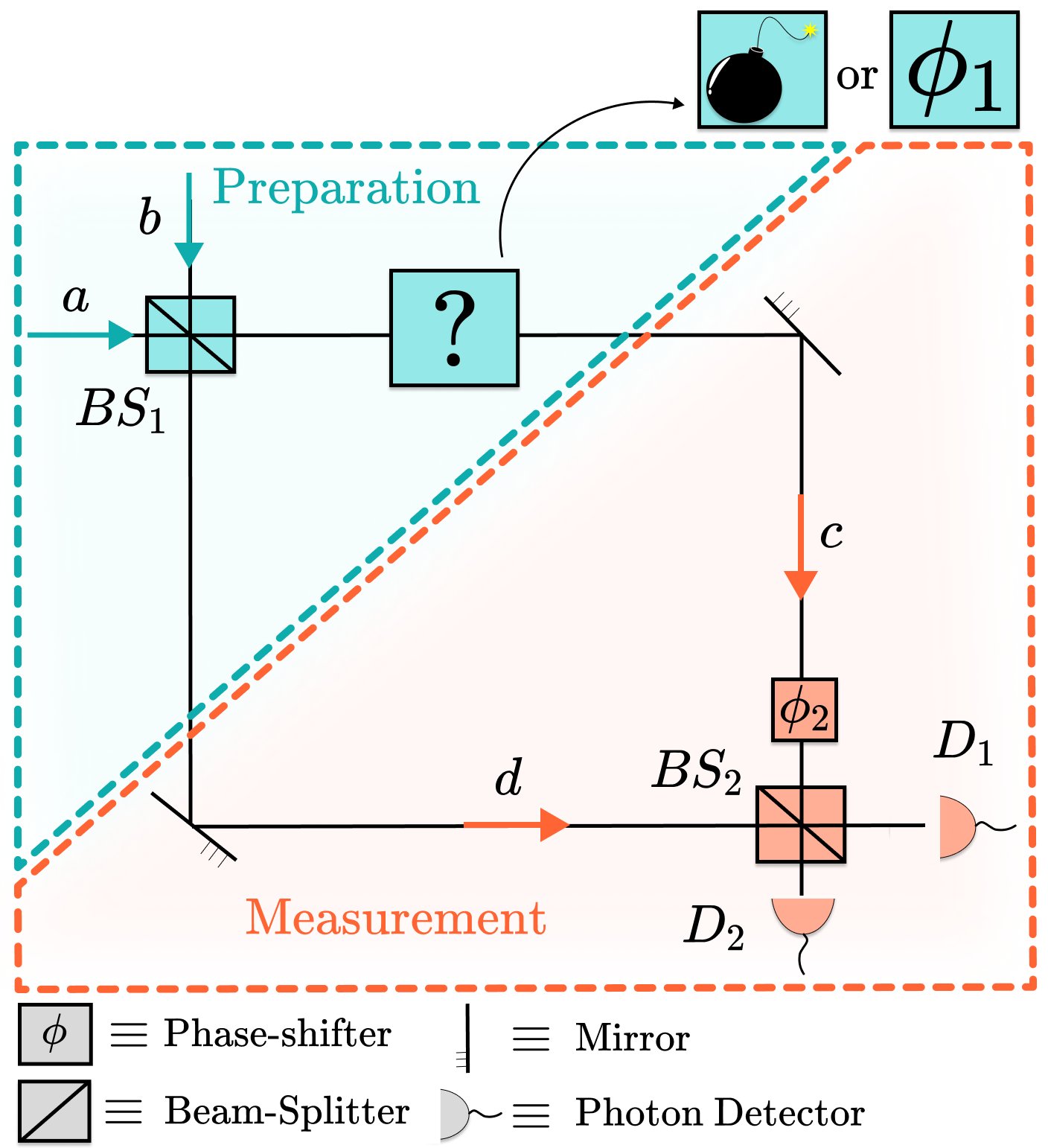}
    \caption{\textbf{Mach-Zehnder Interferometer as a Prepare-and-measure experiment.} In the preparation stage the interferometer is fed with a single photon. The first beam-splitter, $BS_1$ generates superposition between the spatial modes. Qubit path encoding is chosen as $\ket{0}$ and $\ket{1}$ for modes \textit{a} and \textit{b}, respectively. The $?$-box allows for either the presence of a phase-shifter $\phi_1$, placed after $BS_1$, or a bomb that explodes in case it absorbs a photon, present in the standard quantum interrogation scheme. The measurement stage is made of a phase-shifter $\phi_2$, and a second beam-splitter $BS_2$, with photo-detectors placed after each output mode.}
    \label{fig:MZI}
\end{figure}

To build a connection with contextuality scenarios, we can interpret the MZI as a prepare-and-measure device. The MZI apparatus is depicted in Fig.~\ref{fig:MZI}, for the case where the $?$-box represents a PS $U_{\phi_1}$. The two BSs are defined with a fixed $\pi/2$ phase-shift between the optical modes and tunable transmissivity,
\begin{equation}\label{eq: BS}
    U_{\theta_k} := \left(\begin{matrix}\cos\theta_k&i\sin\theta_k\\i\sin\theta_k&\cos\theta_k\end{matrix}\right),
\end{equation}
where $k=1,2$. In this configuration, the \textit{preparation} stage plays the role of a universal one-qubit state generator. Given the orthonormal basis $\{\ket{0},\ket{1}\}$, corresponding to the upper and the lower optical paths respectively, we can encode any pure state using a single photon that enters $BS_1$, for instance in mode \textit{a}:
\begin{align}\label{eq: state inside the interferometer}
    \ket{\psi(\theta_1,\phi_1)} = e^{i\phi_1}\cos\theta_1 \ket{0}+i\sin\theta_1\ket{1}.
\end{align}
Similarly, the \textit{measurement} stage is made of a PS $U_{\phi_2}$ and a BS with photo-detectors at both outputs, a configuration that allows projection onto any qubit state. For instance, if we choose values $\theta_2,\phi_2$ such that  $U_{\phi_1}U_{\phi_2}=\mathbb{1}$ and $U_{\theta_1}U_{\theta_2}=\mathbb{1}$, the measurement stage performs exactly the conjugated operation of the preparation stage, and thus a measurement that exactly corresponds to a projector onto state $\ket{\psi(\theta_1,\phi_1)}$. In the ideal case, detector $D_1$ will click with probability $1 = \vert \langle \psi(\theta_1,\phi_1)\ket{\psi(\theta_1,\phi_1)} \vert^2$, while detector $D_2$ will click with probability $0$. In such a case, the measurement stage performs a dichotomic measurement $M=\{\ket{\psi(\theta_1,\phi_1)}\bra{\psi(\theta_1,\phi_1)}, \mathbb{1}- \ket{\psi(\theta_1,\phi_1)}\bra{\psi(\theta_1,\phi_1)}\}$. It is worth noting that this same measurement $M$ can also perfectly distinguish between the case with an input photon in mode \textit{a}, and the one in mode \textit{b}. If we input the photon in mode \textit{b}, detector $D_1$ will never click, as it is the orthogonal state $\ket{\psi(\theta_1,\phi_1)^\perp}$ the one prepared inside the interferometer.

It is interesting to look at the case where the two BSs are characterized by \textit{different} parameters $\theta$ and $\phi$. By interpreting the measurement stage as a time-reversed one-qubit state generator for a pure state $\ket{\psi(\theta_2,\phi_2)}$, the overall action of the interferometer is to project the state prepared in the first stage onto the state prepared in the second
\begin{align}
    |\bra{0} U_{BS_{2}}U_{\phi_2}U_{\phi_1}U_{BS_1}\ket{0}|^2 =\hspace{0.5cm}\nonumber\\= |\langle \psi(\theta_2,\phi_2)\vert \psi(\theta_1,\phi_1) \rangle|^2.\label{eq: PM MZI as overlaps}
\end{align}
This perspective allows to interpret the MZI as a natural device for estimating quantum two-state overlaps from the frequency of clicks in the detectors $D_1, D_2$, given various choices of PSs and BSs. { As we have seen in Sec.~\ref{subsec: event graphs}, we can infer set-coherence for the collection of states prepared by the MZI using precisely the overlap statistics described by Eq.~\eqref{eq: PM MZI as overlaps}. We will use this precise understanding in the upcoming Secs.~\ref{subsec: parallel setting} and~\ref{subsec: sequential setting} to witness coherence and contextuality.}

One of the several important applications for such a simple linear optics device consists of performing the standard quantum interrogation task~\cite{elitzur1993quantum,vaidman1996interaction}. For this task, in the $?$-box from Fig.~\ref{fig:MZI} is placed an object, typically chosen to be a bomb for historical reasons~\cite{elitzur1993quantum}, that ``explodes'' if and only if it interacts with a single photon. We choose $U_{\theta_1}=U_{\theta_2}^\dagger, \theta_1=\theta_2=\theta$ and $\phi_2=0$, so that the single photon superposition is controlled only by the BSs and one of the detectors is always dark when the bomb is not present. In this scenario, we input the photon in mode \textit{a} (Fig.~\ref{fig:MZI}). The experiment thus consists of testing between two possibilities: \textit{Hypothesis 1) There is a bomb in the left arm of the MZI, but it is not active}, hence it never interacts with the photon,  detector $D_1$ will always click, and detector $D_2$ will never click (usually referred to as the dark detector), or \textit{Hypothesis 2) There is an active bomb in the left arm of the MZI}, and therefore the photon will hit the bomb with probability $\cos^2(\theta)$, detonating it, while it will choose the different path with probability $\sin^2(\theta)$. The latter case, corresponding to a state $\ket{1}$ after the first BS, gives a chance to detect the object/bomb without exploding it, exploiting the bomb as a complete path-information measurement device. In fact, after $\vert 1 \rangle$ passes the second BS, detector $D_1$ will click with probability $\sin^2(\theta)$, and detector $D_2$ will click with probability $\cos^2(\theta)$. However, $D_2$ clicks \textit{only} in the case of an unexploded active bomb, thus, detecting the presence of the object. With this protocol in mind we define the following information task:
\begin{tcolorbox}[colback=blue!6,colframe=blue!35,left=1mm,right=1mm,top=1mm,bottom=1mm,arc=0mm]
\textit{\textbf{Quantum interrogation:}} Using as many photons as needed, detect the presence of an active bomb \textit{without} exploding it, with the highest possible probability.
\end{tcolorbox}

If the BSs are symmetric, detector $D_2$ clicks with probability $1/4$. Therefore we may need to run the experiment many times to see any $D_2$ event, having the drawback that, by doing so, eventually the bomb will explode. The figure of merit for the efficiency of this task $\eta$ is defined operationally as~\cite{elitzur1993quantum,kwiat1995interactionfree},
\begin{equation}\label{eq: efficiency}
    \eta = \frac{p_{succ}}{p_{succ} + p_{bomb}}
\end{equation}
where $p_{succ}$ corresponds to the probability that the dark detector $D_2$ clicks and we successfully detect the bomb without it exploding, while with probability $p_{bomb}$ the bomb explodes. In case of symmetric BSs, as pointed out in Ref.~\cite{elitzur1993quantum}, we have $\eta = \frac{1/4}{1/4+1/2} = \frac{1}{3}$. The same efficiency can be achieved by the noncontextual model of Ref.~\cite{catani2021interference}.

\section{Results}

The event graph formalism for witnessing basis-independent coherence helps quantify non-classical superpositions from two different perspectives. We call these two perspectives the \textit{parallel setting} and the \textit{sequential setting}, as schematically depicted in Fig.~\ref{fig: parallel vs sequential}.

\subsection{Parallel setting}\label{subsec: parallel setting}
\begin{figure}[tb]
    \centering
    \includegraphics[width=\columnwidth]{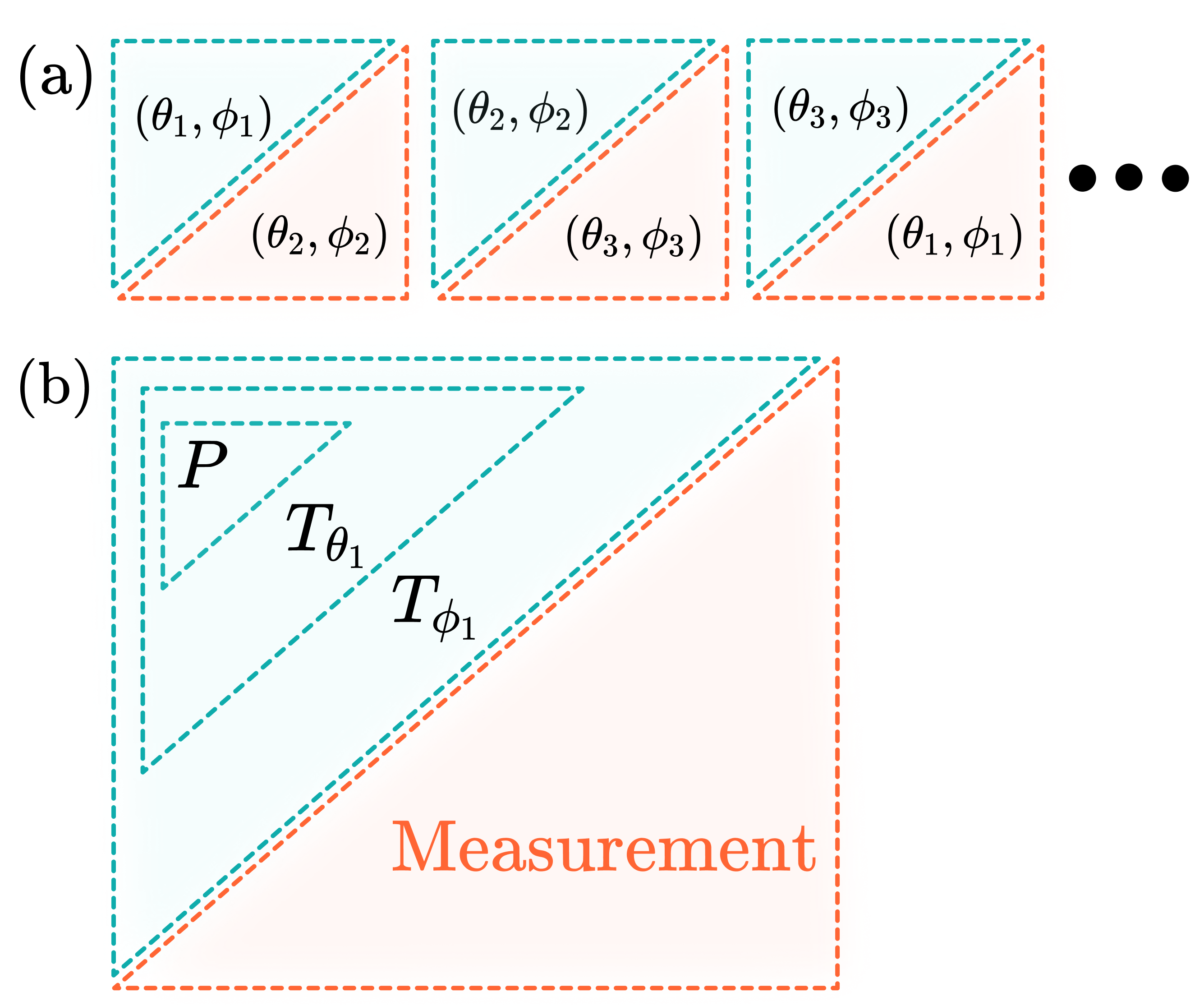}
    \caption{\textbf{Parallel versus sequential settings.} We follow the prepare-and-measure scheme for representing the Mach-Zehnder interferometer as in Fig.~\ref{fig:MZI}. (a) We choose a finite number of MZI configurations to probe coherence of quantum states inside the interferometer. Measurements correspond to dual representations of the chosen quantum states. (b) We interpret a single MZI configuration as capable of preparing different states due to BS and PS transformations. The measurement stage performs tomographically complete measurements.}
    \label{fig: parallel vs sequential}
\end{figure}

In a parallel setting, shown in Fig.~\ref{fig: parallel vs sequential}-(a), we investigate coherence as a relational property of the several different states that can be generated inside the MZI, for various choices of BSs and PSs.
The coherence of the states $\vert \psi(\theta_i,\phi_i) \rangle $ inside the interferometer can be witnessed and quantified by measuring only a finite number of two-state overlaps. This procedure is capable of evidencing basis-independent coherence, a stronger form than standard notions of basis-dependent coherence, using only a small sample of chosen quantum states inside the MZI. 

{ Experimentally, probing coherence using violation of the inequalities reviewed in Section~\ref{subsec: event graphs} with a parallel setting amounts to performing various prepare-and-measure experiments, each estimating a given overlap $r_e$. In such a way, we can witness coherence solely with the overlap statistics, without assuming a fixed sequence of transformations as in the sequential setting, to be discussed in the next section. Moreover, depending on the inequality considered, one can significantly simplify the procedure for witnessing coherence, by using some finite number of phase-shifter settings, unlike methods based on the visibility of an interference fringe, requiring a continuous set of measurements. Finally, this protocol can be extended in a rather straightforward manner to probe coherence in multimode interferometric devices, hence beyond the two-mode MZI interference set-up we focus on here. 
}

Figure~\ref{fig:violation}-(a) shows how witnessing coherence can be done with asymmetric BSs, using the $C_3$ inequalities of Eqs.~\eqref{eq: c3 ineq 1}-\eqref{eq: c3 ineq 3}. One can also probe coherence using inequality \eqref{eq: inequality (d)}, by setting a single symmetric BS, and choosing a  set of $5$ phases $\phi$ { for preparations, and again 5 for measurement states}, or equivalently $10$ effective phase-differences, { each for the estimation of a single overlap $r_{ij}$ from Eq.~\eqref{eq: inequality (d)},} that allow for qubit violations up to $\frac{5\sqrt{5}}{4}\sim 0.79$~\cite{wagner2022inequalities}. Figure~\ref{fig:violation}-(c) reports five different choices of phases $\{\phi_i\}_{i=1}^5$ to witness set coherence for the states $\{\vert \psi(\pi/4,\phi_i)\rangle \}_{i=1}^5$ in the form of Eq.~\eqref{eq: state inside the interferometer}, showing that in this case we can find high violations with only five well-chosen PS settings.
\begin{figure*}[htb]
    \includegraphics[width=1\textwidth]{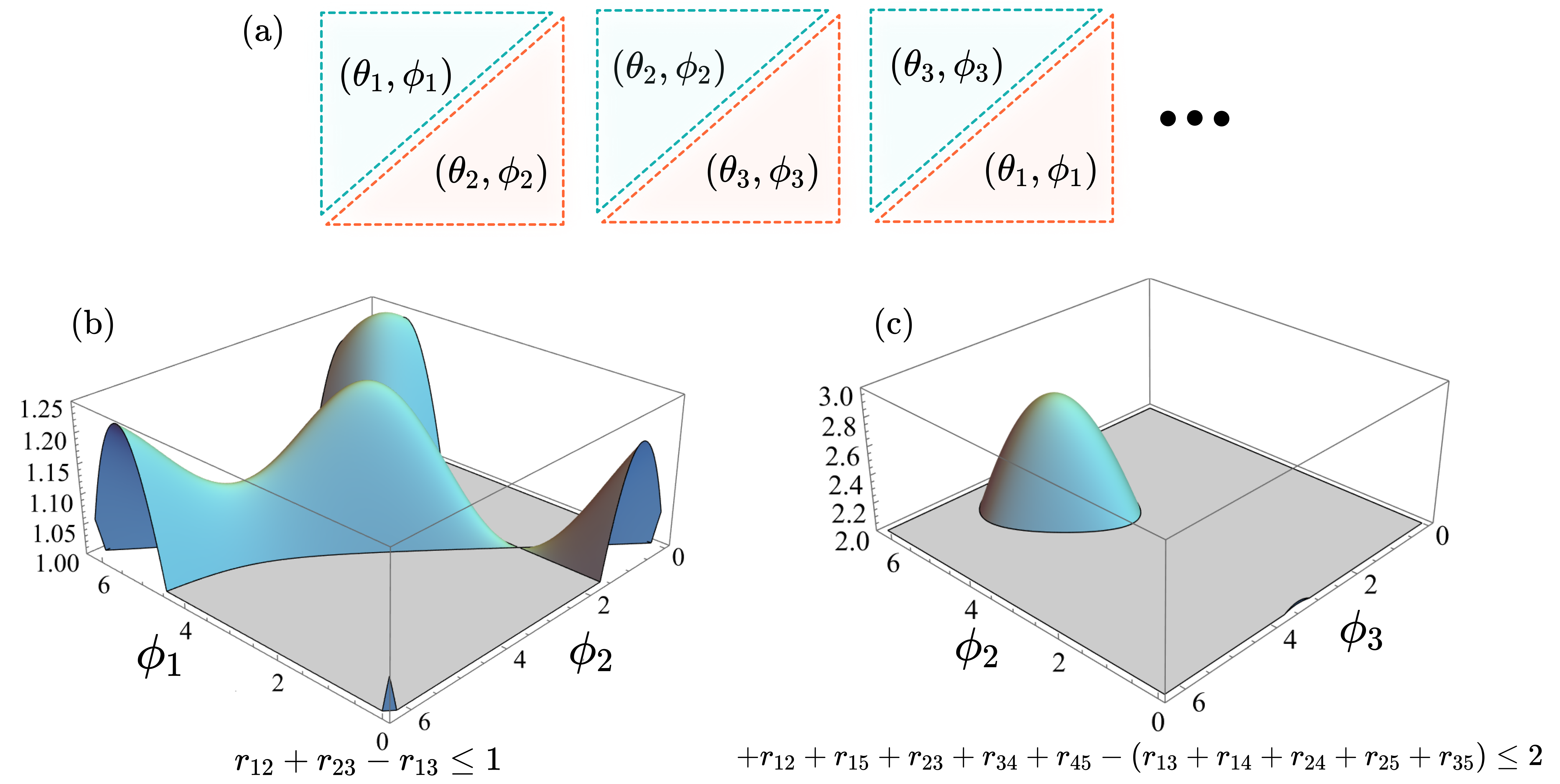}
    \caption{\textbf{Witnesses of basis-independent coherence and preparation contextuality in a standard Mach-Zehnder Interferometer (MZI).} Plot of phases $\phi_i$ and beam splitter's transition/reflection rates $\theta_i$ for which there are violations. A parallel setting is shown in (a),  where several runs of a given MZI allow to estimate the overlaps of the chosen states. The states inside the MZI are taken to be in their general form as $\vert \psi(\theta_i,\phi_i) \rangle = e^{i\phi_i}\cos(\theta_i)\vert 0 \rangle + i \sin(\theta_i)\vert 1 \rangle$. In plot (b), three different MZI configurations are considered, with both symmetric and asymmetric BSs, and a fixed PS in the third state; the chosen states are $\{\vert \psi(\theta_1 = \pi/4,\phi_1)\rangle , \vert \psi(\theta_2 = \pi/3,\phi_2)\rangle , \vert \psi(\theta_3 = 4\pi/4,\phi_3 = 6\pi/3) \rangle \}$, and the plot shows many possible violations of the $C_3$ inequality. In plot (c), the BSs are all taken to be symmetric, i.e., $\theta_i = \pi/4$ for all states $i=1,\dots,5$; in the plot, three out of five phase choices are fixed to be $\phi_1=0, \phi_4=4\pi/5, \phi = \pi/3$, letting the two remaining phases $\phi_2,\phi_3$ to vary in $[0; 2\pi]$, in order to look for high violations. Pure state violations in (b) can be mapped to contextual behaviors in prepare-and-measure scenarios, and general states violating (b) or (c) are basis-independent coherence witnesses.}  
    \label{fig:violation}
\end{figure*}

\subsubsection{Comparison with other methods for quantifying coherence}
The parallel setting scheme performs better over the use of visibility as a measure of nonclassicality for the states inside the interferometer. Estimating the visibility requires a maximization process over all possible values of $\phi$, whereas our scheme only needs to probe finitely many chosen values of $\phi$. Moreover, the visibility has no direct link with stronger forms of nonclassicality as our witnesses do have~\footnote{The authors of Ref.~\cite{catani2021interference} are working in solving this precise issue by providing clear cases for when the visibility cannot be explained with generalized noncontextual models.}, and it is not clear if it could potentially be used to witness coherence in a basis-independent form. Not only that, the visibility has been criticized as a measure of nonclassical behaviour inside the interferometer~\cite{chrysosthemos_quantum_2022}. Recent literature points towards the realization that resource theoretic coherence monotones quantify more consistently the non-classicality of quantum states inside interferometers.

We may hence compare our approach with known results in the basis-dependent coherence framework~\cite{baumgratz2014quantifying,wu2021experimental}, beyond the connection with contextuality. Coherence \textit{monotones}~\cite{Streltsov17,wu2021experimental} have been proposed as a way to quantify coherence. Monotones have been widely studied in the abstract formalism, and are elements in complementarity relations~\cite{Streltsov17,wu2021experimental}. Experimentally probing the majority of interesting coherence monotones - as they were originally proposed - requires full quantum state tomography. Although tomography is a fairly standard experimental construction that is constantly being improved~\cite{rambach2021robust}, the number of measurements and samples necessary for this task grows exponentially with the number of qubits in the best possible case, with or without the help of full quantum processing power~\cite{chen2022tight}. 

A way to avoid state tomography for quantifying coherence is via experimentally probing \textit{coherence witnesses}~\cite{zhang2018estimating,wu2021experimental,napoli2016robustness,wang2017directly}, which greatly improve experimental investigations of coherence monotones. Nevertheless, there is the draw-back that, even if finding the specific witness that provides the value of the monotone could be simple by using semi-definite programming, the actual implementation of this observable may be difficult as the dimension of the system increases~\cite{zheng2018experimental}. 

Notably, as it is clear from Sec.~\ref{subsec: event graphs}, this framework is capable of witnessing coherence independently of the way one encodes quantum information in the degrees of freedom of light, i.e., independently of Hilbert space dimension. Using the degree of violation of the quantum states as a generic measure of coherence of the set of states, as in Ref.~\cite{giordani2020experimental}, the lack of need for state tomography, or for maximization procedures, indicates that scalability of our approach is itself a promising tool for investigating multimode quantum interference. { The parallel setting just described can be used to certify coherence in complex interferometric settings beyond the MZI, of the kind that have been shown to be useful for quantum computation with linear optics~\cite{taballione2021universal}.}

\subsection{Sequential setting}\label{subsec: sequential setting}

In the sequential setting, see Fig.~\ref{fig: parallel vs sequential}-(b), we study the ability of a single choice of MZI parameters, i.e. fixed values of $\theta_1$ and $\phi_1$, of generating statistics that cannot be properly explained by a noncontextual model. For such, we consider three states $\vert \psi(0,0) \rangle$, $\vert \psi(\theta_1,0)\rangle$, $\vert \psi(\theta_1,\phi_1)\rangle$. Operationally, they correspond to the preparation of the input state $\vert 0 \rangle$, followed by the unitary describing the beam-splitter, $U_{\theta_1}$, and then by the one describing the phase-shifter, $U_{\phi_1}$. We show in Fig.~\ref{fig:BS fixed violations}-(b) that there are choices of $(\theta_1,\phi_1)$ implying violations of the $C_3$ inequalities { in such a sequential description}, which as discussed in Ref.~\cite{wagner2022inequalities} represent witnesses of contextuality.

{The main relevance of the sequential setting we discuss in this section might be its connection with contextuality, to be formally described in Sec.~\ref{subsubsec: gen contextuality in sequential}, and which will be instrumental in Sec.~\ref{subsubsec: contextual advantage} to prove a contextual advantage for the task of quantum interrogation. Before delving into the connection between event graph inequalities and contextuality inequalities we must first clarify some subtleties.}

{We have seen in Section~\ref{subsec: generalized contextuality} that one can discuss noncontextuality defined for prepare-and-measure scenarios. As discussed there, the set of noncontextual behaviors $\mathcal{NC}$ is completely characterized by inequalities defined for behaviours. Similarly, we have introduced in Section~\ref{subsec: event graphs} an inequality framework bounding a different set, denoted $C_G$, constituting two-state overlap inequalities. A priori, there is no clear connection between the two frameworks. 
For example, it is clear that behaviours $\mathbf{p}$ are more general than two-state overlaps, because they allow for generic statistics of the form $p(k|M,P) = \text{Tr}(\rho E_k)$ while for overlaps we require that $E_k = \sigma $ (some quantum state). }

{In the rest of this section we show that each $3$-cycle overlap inequality discussed in Section~\ref{subsec: event graphs} can be interpreted as a noncontextuality inequality from a prepare-and-measure scenario from the framework introduced in Section~\ref{subsec: generalized contextuality}. For each $3$-cycle overlap inequality, we will explicitly construct a prepare-and-measure scenario, with the structure reviewed in Sec.~\ref{subsec: generalized contextuality}, such that the chosen overlap inequality is now a valid noncontextuality inequality, for the constructed scenario. A similar approach was taken in Ref.~\cite{wagner2022inequalities}, but assuming ideal preparations and measurements. We generalize this connection, and present the  novel robust noncontextuality inequalities described in Eq.~\eqref{eq: robust c3 inequality}. Importantly, we then show that MZIs satisfy the exact operational constraints of the scenario constructed, implying that it is possible to test the noncontextuality inequalities found using such a simple experimental setting. Later, we use this inequality to quantify contextual advantage for quantum interrogation in Sec.~\ref{subsubsec: contextual advantage}. }

Differently from the case of basis-independent coherence, the connection with contextuality must be made rigorous by connecting these states and the experimental set-up of an MZI with some prepare-and-measure contextuality scenario~\footnote{Note that the sequential setting is an instance of  the parallel setting, since we can generate the triples of states just discussed, as well as all their overlaps, using parallel setting only. We hereby focus on the interpretation that can be drawn by possible noncontextual explanations of a single MZI interferometer.}, as we show now. 

\subsubsection{Generalized contextuality in the sequential setting}\label{subsubsec: gen contextuality in sequential}

It is simple to translate the statistics arising from event graphs as a subset of the statistics arising from prepare-and-measure operational scenarios, in a way that event graph cycle inequalities become noncontextuality inequalities for such scenarios. This operational translation arises from structural results that noncontextual quantum two-state overlaps must obey from Refs.~\cite{Lostaglio2020contextualadvantage,schmid2018discrimination}, together with self-duality property of quantum theory~\cite{janotta2013generalized,muller2012structure}, for which any state can also be viewed as a measurement effect, and from the fact that the operational statistics of those scenarios can be related to the edge-assignments $r$ of event graphs $G$. 

Let us now describe the relation between the $C_3$ event graph, and a prepare-and-measure scenario for the specific case of the MZI experiment we are interested in, as depicted in Fig~\ref{fig:From Event Graph to Operational Scenarios}. We discuss the generality of this approach in Appendix B. 


Inside the MZI, two transformations $T_{\theta_1}$ and $T_{\phi_1}$ applied in sequence follow the preparation $P_0$ of a given state. The preparation $P_0$ corresponds to an element of the class of equivalent procedures represented by a state $\vert 0 \rangle \langle 0 \vert$ in a prepare-and-measure setting, see Fig.~\ref{fig: parallel vs sequential}-(b). Each \textit{fixed} pair of transformations $T_{\theta_1}$ and $T_{\phi_1}$ operationally defines new preparation procedures $T_{\theta_1}(P_0)$ and $T_{\phi_1}(P_0)$. We let $P_1 = T_{\theta_1}(P_0)$ and $P_2 = T_{\phi_1}(P_1)$, corresponding to only one of the possible preparations that can be performed in the chosen scenario. Operationally we also assume that there exist preparations $P_{0^\perp},P_{1^\perp},P_{2^\perp}$ such that $\frac{1}{2}P_0+\frac{1}{2}P_{0^\perp}\simeq \frac{1}{2}P_1+\frac{1}{2}P_{1^\perp} \simeq \frac{1}{2}P_{2}+\frac{1}{2}P_{2^\perp}$, and we let the set $\mathcal{OE}_P^{MZI}$ have this operational equivalence. {As will be clear later when we discuss the quantum implementation with the MZI, for any $i$, $P_{i^\perp}$ corresponds to preparing the state orthogonal to the one associated with preparation $P_i$.} It is well known that behaviours $\mathbf{p}$ in scenarios satisfying similar operational equivalences are related to behaviours in Bell scenarios~\cite{chaturvedi2021characterising}. However, we make no assumptions related to probe non-locality for the inequalities we consider in order to claim they are Bell inequalities, since we are interested in probing generalized contextuality instead of nonclassicality of common causes, i.e., Bell nonlocality.

The final ingredient of the prepare-and-measure scenario must be the measurement effects. We can set binary-outcome measurements $M_1, M_2, M_3$ with the constraints $p(0\vert M_i,P_i)\geq 1-\varepsilon_i$, and $p(0 \vert M_i, P_{i^\perp})\leq \varepsilon_i$ $\forall i \in \{1,2,3\}$. In the ideal case we have $\varepsilon_i=0$, and the inequalities become equalities. Under such assumptions, each $p(0\vert M_j,P_i)\equiv p(M_j\vert P_i)$ describes an operational \textit{confusability}~\cite{schmid2018discrimination}, because it allows for the interpretation as the probability that process $P_j$ passes the test of being process $P_i$, under the discrimination measurement $M_i$. We can then write the scenario $\mathbb{B} = (6,3,2,\mathcal{OE}_P^{MZI},\emptyset)$. We will refer to these operational scenarios as the Lostaglio-Senno-Schmid-Spekkens scenarios (LSSS)~\cite{Lostaglio2020contextualadvantage,schmid2018discrimination}, as well as to the general constructions made in Appendix B. Figure~\ref{fig:From Event Graph to Operational Scenarios} depicts the connection between $C_3$ and $\mathbb{B}$. Edges correspond to operational probabilities $p(0 \vert M_i, P_j)$. Green edges have the constraints  $p(0\vert M_i ,P_i)=1$, and edges with $p(0\vert M_i ,P_{i^\perp})=0$ are not represented. We thereafter note that we can interpret the confusabilities $r_{ij} = p(0\vert M_i, P_j)$ from the imposed constraints. Structurally, any ontological model for the behaviours of scenario $\mathbb{B}$ with the constrains described satisfies the following inequality~\cite{Lostaglio2020contextualadvantage}
\begin{equation}\label{eq: LSSS structure}
    \vert \Vert \mu_i - \mu_{i+1}\Vert - 2(1-r_{i i+1}) \vert \leq 2\varepsilon_i ,
\end{equation}
where summation in the labels is mod 3, the epistemic states $\{\mu_i\}_{i=1}^3$ are the ontological model representation for the preparations $\{P_i\}_{i=1}^3$, and the norm corresponds to the $l_1$-norm defined over the space $\Lambda$, given by $\Vert \mu_i - \mu_j \Vert := \int_\Lambda \vert \mu_i(\lambda) - \mu_j(\lambda)\vert d\lambda$. 

{In Figure~\ref{fig:From Event Graph to Operational Scenarios} we present a schematic translation between the event graph $C_3$ and the prepare-and-measure scenario we investigate. Note that, as is done throughout our work, we label the nodes of event graphs as $1,2,\dots,n$. In the figure, we initially show $C_3$ as before in Fig.~\ref{fig: simplest event graph}. Each node $i$ in the event graph is then mapped to a preparation node $P_{i-1}$, a measurement effect node $0|M_{i-1}$, and another preparation node $P_{i^\perp}$. All these operational constituents, together with the operational equivalence shown, form the scenario $(6,3,2,\mathcal{OE}_P^{MZI},\emptyset)$ mentioned before. For the interferometric description that follows, it is interesting to make this slight change in labels $i \mapsto i-1$ from event graph to the prepare-and-measure primitives. In this way, for Fig.~\ref{fig:From Event Graph to Operational Scenarios}, we have $$r_{ij} = p(0|M_{i-1},P_{j-1}) = p(0|M_{j-1}, P_{i-1}).$$ This is done so that we can naturally relate $P_0$ with preparation of $\vert 0\rangle \langle 0 \vert$. Preparation and measurement procedures alternate in an inner graph. We also include \emph{new} procedures $P_{0^\perp}, P_{1^\perp}, P_{2^\perp}$ corresponding to the included operational elements satisfying the operational equivalence in the figure, characteristic of the scenario~\footnote{Note that, from a purely graph theoretic perspective, the transformation between the two graphs is non-trivial. Yet, operationally, the rules for constructing the prepare-and-measure scenario described here, and also in Appendix~\ref{appendix: general}, are straightforward.}. From this notation, we naturally associate $P_{0^\perp}$ with the preparation of $\vert 1\rangle \langle 1 \vert$. }

As a result from the LSSS structural results of Eq.~\eqref{eq: LSSS structure}, we find that the {3-cycle inequalities, facet-defining for the convex polytope  $C_{C_3}$ associated with event graph $C_3$ (c.f. Section~\ref{subsec: event graphs}), are} actually related to the robust noncontextuality inequality for the LSSS scenario just described and depicted in Fig.~\ref{fig:From Event Graph to Operational Scenarios},
\begin{equation}\label{eq: robust c3 inequality}
    r_{12}+r_{13}-r_{23} \leq 1 + \varepsilon_1+\varepsilon_2+\varepsilon_3.
\end{equation}
which can be shown by a simple application of the triangle inequality. A similar result holds for any $n$-cycle overlap inequality, which allows us to prescribe a new family of robust noncontextuality inequalities. We also remark that a similar construction using the $C_3$ noncontextuality inequality structure has appeared in Ref.~\cite{flatt2021contextual}.

\begin{figure}[tb]
    \centering
    \includegraphics[width=\columnwidth]{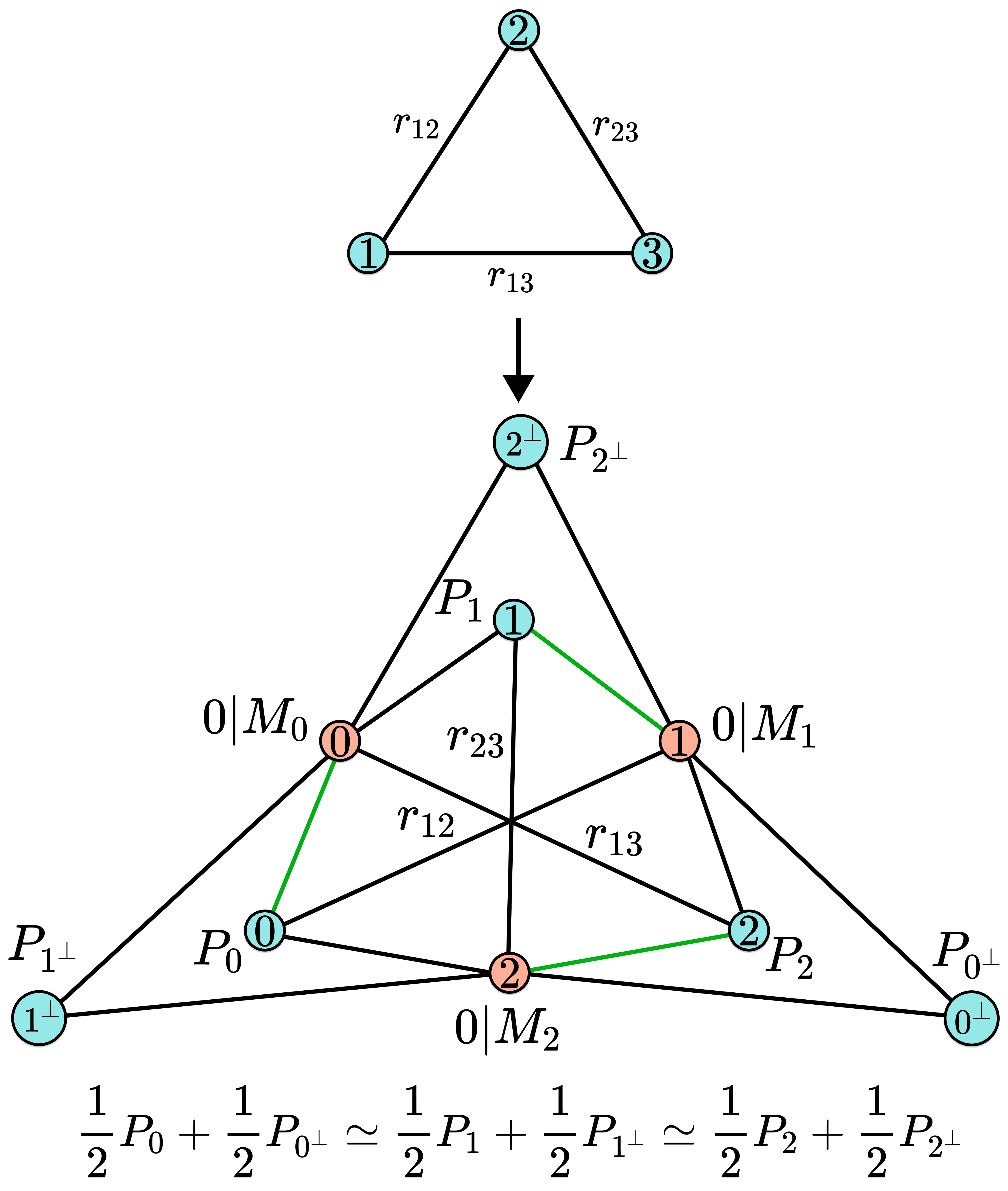}
    \caption{\textbf{Contruction of the prepare-and-measure scenario from the $C_3$ graph.} The event graph with edges $r_{12},r_{13},r_{23}$ is mapped towards a prepare-and-measure scenario where each vertex is associated to two new preparation procedures and one measurement effect. Edge-assignments in the latter are mapped to symmetric confusabilities in the former.}
    \label{fig:From Event Graph to Operational Scenarios}
\end{figure}

Although the noncontextual confusabilities and the coherence-free two-state overlaps differ in interpretation, their description collapses in the deterministic assignments defining the classical $C_{C_3}$ polytope, implying that the two seemingly different interpretations can be linked by the same inequalities. We map event graph inequalities into noncontextuality inequalities by mapping edge-assignments $r_{ij}$ into confusabilities satisfying $r_{ij}=r_{ji}$.

It remains to show that the statistics arising from the quantum theory of the MZI \textit{respect} the operational equivalences and symmetries of the LSSS scenario $\mathbb{B}$. When we treat quantum theory as an operational theory, the confusabilities $p(M_i \vert P_j)$ match the two-state overlaps related to processes $P_i$ in the equivalence class defined by $\vert \phi_i \rangle$. In this sense, we have that $r_{ij} = p(M_i \vert P_j) = \vert \langle \psi_i \vert \psi_j \rangle \vert^2 = p(M_j \vert P_i) = r_{ji}$, in the ideal case where we consider pure quantum states as vertices of the $C_3$ graph. This connection also clarifies that the violation of generalized inequalities needs basis-independent coherence of the chosen states, that \emph{per se} may be seen as an application of basis-independent coherence in the context of strong witnesses of contextuality. In the case of the MZI in Fig.~\ref{fig:MZI}, the preparation procedures $P_0$ and $P_1$ correspond to {preparing} a photon in modes \textit{a} or \textit{b}, states $\vert 0 \rangle \langle 0 \vert$ and $\vert 1 \rangle \langle 1 \vert$, respectively. The transformations are described by unitaries $U_{BS_1}(\theta_1)$ and $U_{\phi_1}$, which imply that
\begin{align*}
    \mathbb{1}/2 &= \frac{1}{2}\vert 0 \rangle \langle 0 \vert + \frac{1}{2}\vert 1\rangle \langle 1 \vert =\\ &=\frac{1}{2}U_{BS_1}(\theta_1)\vert 0 \rangle \langle 0 \vert U_{BS_1}(\theta_1)^\dagger+\\ &\hspace{4.5em}+\frac{1}{2}U_{BS_1}(\theta_1)\vert 1\rangle \langle 1 \vert U_{BS_1}(\theta_1)^\dagger =\\ &=\frac{1}{2}U_{\phi_1}U_{BS_1}(\theta_1)\vert 0 \rangle \langle 0 \vert U_{BS_1}(\theta_1)^\dagger U_{\phi_1}^\dagger +\\
    &\hspace{4.5em}+\frac{1}{2}U_{\phi_1}U_{BS_1}(\theta_1)\vert 1\rangle \langle 1 \vert U_{BS_1}(\theta_1)^\dagger U_{\phi_1}^\dagger.
\end{align*}
Hence, the preparation of the states $\vert \psi(0,0) \rangle, \vert \psi(\theta_1,0) \rangle$, $\vert \psi(\theta_1,\phi_1) \rangle$ and of their complements satisfies the operational equivalences we have described before. Moreover, the MZI's measurement stage can universally implement the binary outcome measurements $M_1 = \{\vert \psi(0,0) \rangle \langle \psi(0,0) \vert, \mathbb{1} -\vert \psi(0,0) \rangle \langle \psi(0,0) \vert \}$, as well as for the remaining states. Defining the quantum effects {$[0 \vert M]$} to be the same as the quantum states being prepared, we have ideally the signatures $p(M_1\vert P_1) = 1$ and $p(M_1 \vert P_{1^\perp}) = 0$. 
With the tools provided so far, let us now describe \textit{when} one can witness generalized contextuality in a single MZI, using the simplest $3$-cycle inequality. {Considering the states $\vert \psi_1 \rangle = \vert \psi(0,0)\rangle, \vert \psi_2 \rangle = \vert \psi(\theta_1,0)\rangle$ and $\vert \psi_3 \rangle = \vert \psi(\theta_1,\phi_1)\rangle$,} we can calculate { Eq.~\eqref{eq: c3 ineq 1} 
\begin{equation}\label{eq: c3 functional for standard MZI}
    h(\theta_1,\phi_1) := \vert \langle \psi_1 \vert \psi_2 \rangle \vert^2 + \vert \langle \psi_1 \vert \psi_3 \rangle\vert^2 - \vert \langle \psi_2 \vert \psi_3 \rangle\vert^2 \leq 1,
\end{equation}}
which will characterize the violations with respect to this encoding. Note that this is the only functional $h$, of all $3$ non-trivial~\footnote{Note that these functionals define the facet inequalities from the polytope $C_{C_3}$, which also has \textit{trivial} inequalities of the form $0 \leq r_e \leq 1$, for all edges $e \in C_3$. } ones { presented in Eqs.~\eqref{eq: c3 ineq 1}-\eqref{eq: c3 ineq 3} from the convex polytope defined with respect to } the $C_3$ graph, that can have a quantum violation, due to the sequential setting encoding of states we choose. Figure~\ref{fig: violations for any theta and phi}-(b) shows all the possible choices of BSs and PSs to violate the contextuality inequalities of the $C_3$ event graph given the constraint imposed by the choice of triple of states just described. 
\begin{figure}[tb]
    \centering
    \includegraphics[width=\columnwidth]{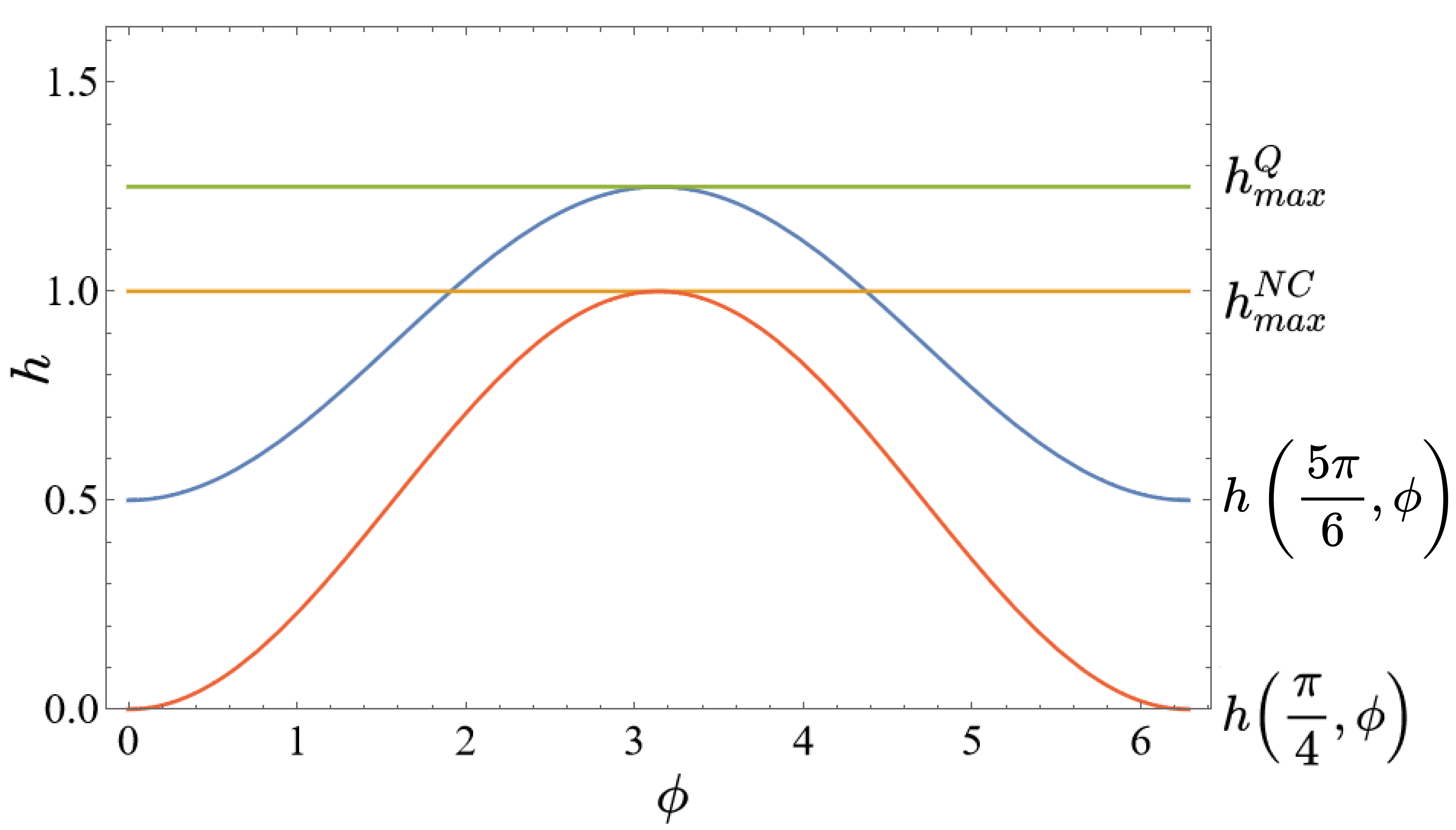}
    \caption{\textbf{Violations for fixed beam-splitter (BS) ratios.} The lower curve shows the case where the first BS in the interferometer is symmetric, i.e. $\theta=\pi/4$, and it is clear that it is never possible to witness contextuality with the inequality \eqref{eq: c3 ineq 1}. On the opposite, the case with $\theta=\pi/3$ allows the best violation for variations in the phase parameter $\phi$. In the picture, due to a choice of encoding in the description of the first BS, the phase for this latter case is taken to be {$\theta=5\pi/6=\pi/2+\pi/3$}.}
    \label{fig:BS fixed violations}
\end{figure}
When the first BS is symmetric, corresponding to $\theta=\pi/4$, we find \textit{no} violation of the noncontextuality inequality $h(\theta_1,\phi_1) \stackrel{NC}{\leq} 1$, see also Fig.~\ref{fig:BS fixed violations}. This is consistent with Ref.~\cite{catani2021interference}, that explicitly provides a noncontextual model for the cases with $\theta = \pi/4$ and $\phi \in \{0,\pi\}$, from modifications of the famous Spekkens Toy Model of Ref.~\cite{spekkens2007evidence}. As shown in Fig.~\ref{fig:BS fixed violations}, if we relax the constraint of having the first BS to be symmetric, for $\theta=\pi/3$ we find maximal violation for quantum states over a great circle in the Bloch sphere, consistently with the results in Refs.~\cite{galvaobroad2020quantumandclassical,giordani2021witnessesofcoherence}. As a remark, we notice that not all choices made for BSs and PSs violate the $C_3$ inequalities we have chosen, but the related states could still be used to violate \textit{some other} noncontextuality inequality different than the one we are focused {on.} 
\begin{figure*}[ht]
    \includegraphics[width=1\textwidth]{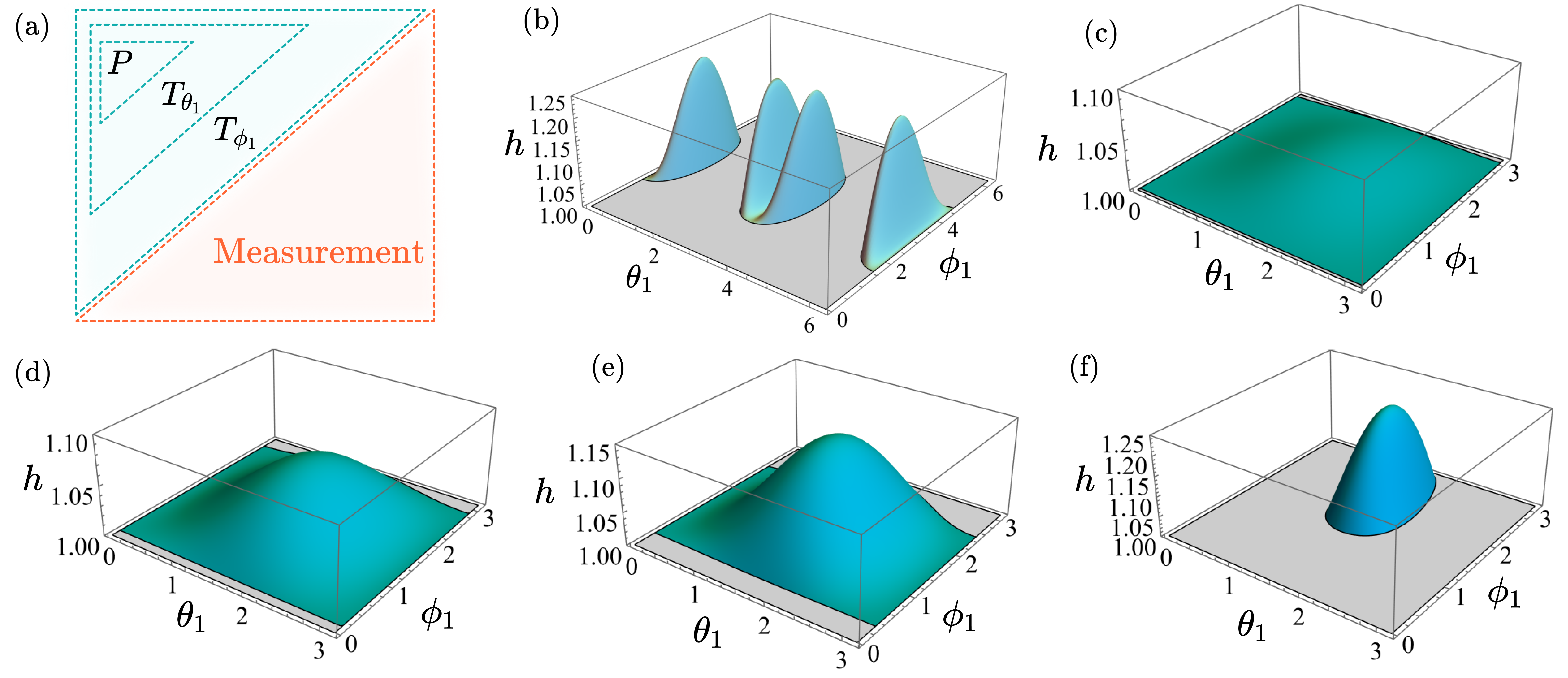}
    \caption{\textbf{MZI violations for the states related to the sequential setting.} Plots of the functional $h$ from Eq.~\eqref{eq: c3 functional for standard MZI} for different triples of states, analysing the dynamics of a single MZI in the sequential setting (a). Plot (b) shows violations for the states    $\{\vert 0 \rangle, \cos(\theta_1)\vert 0\rangle + i\sin(\theta_1)\vert 1\rangle, e^{i\phi_1}\cos(\theta_1)\vert 0\rangle + i\sin(\theta_1)\vert 1\rangle\}$, and plots (c)-(f) show violations for the states $\{\frac{1}{\sqrt{2}}(\vert 0\rangle +ie^{i\phi}\vert 1\rangle), U_{\theta_1}\frac{1}{\sqrt{2}}(\vert 0\rangle +ie^{i\phi}\vert 1\rangle), P_{\phi_1}U_{\theta_1}\frac{1}{\sqrt{2}}(\vert 0\rangle +ie^{i\phi}\vert 1\rangle)\}$, being $P_{\phi_1}$ the phase-gate, and $U_{\theta_i}$ the beam-splitter. While (b) is given by a fixed input state $\vert 0 \rangle$, (c)-(f) have different superpositions of $\vert 0 \rangle$ and $\vert 1 \rangle$ as input. In this latter case, some input states provide violation for virtually \textit{any} MZI configuration. In figures (c)-(f), the phase varies as $\phi = \pi/2 + \pi/k$ for $k=150,50,20,5$, respectively. For values of $\phi$ arbitrarily close but different from $\pi/2$, as in plots (c)-(e), practically any pair $(\theta,\phi) \in [0,\pi]^2$ provides a violation. { Concretely, in (c), every pair $(\theta_1,\phi_1) \in [0,\pi]\times [0.042,\pi-0.084]$ provides a violation.} Experimental precision in evaluating overlaps via parallel SWAP-test using photons currently achieve an order of $\approx 0.01$ and confidence of more than $5$ standard deviations~\cite{giordani2021witnessesofcoherence}. }
    \label{fig: violations for any theta and phi}
\end{figure*}

From the results presented so far it is clear that, in order to access violations of the simple $C_3$ inequalities of a fixed MZI setting, we need the input photons to pass through an \textit{asymmetric} BS. The maximal value $h_{max}^Q$ of a violation is given by $1/4$, and can be experimentally investigated in a robust way. This maximal violation was analytically obtained in Ref.~\cite{galvaobroad2020quantumandclassical}. These results allow to interpret the pairs $(\theta,\phi)$ not only as coherence generators~\cite{ares2022beam}, but also as preparation contextuality generators, again quantified by the degree of violation. From the robust inequality in Eq.~\eqref{eq: robust c3 inequality}, we find a threshold for the allowed noise parameters $\{\varepsilon_i\}_{i=1}^3$: if their sum surpasses $1/4$, no quantum violation can be found.

Up to now, we studied the case for which the input of the MZI in Fig.~\ref{fig:MZI} consisted of the state $\ket{0}\bra{0}$ (or equivalently $\ket{1}\bra{1}$). We can relax this constraint by using \textit{generic} input states, and understand what the benefits may be of doing so. One possibility is to simply prepare general states using additional BS and PS before entering the MZI~\cite{chrysosthemos_quantum_2022}, to prepare general two-mode superpositions that are then fed into the input ports of the  MZI. Strikingly, it is in fact possible to find specific input states $\ket{\psi}$ that violate the $C_3$ inequalities for {\textit{virtually all states}} inside the  MZI. {More precisely, it is possible to find input states $\vert \psi \rangle$ so that, as we can see numerically, for any choice of BS $U_{\theta_1}$ and any choice of PS $U_{\phi_1}$ we have that,
\begin{equation*}
    \{\vert \psi \rangle, U_{\theta_1}\vert \psi \rangle, U_{\phi_1}U_{\theta_1}\vert \psi \rangle \}
\end{equation*}
violates the $3$-cycle inequality considered. In Appendix C, we calculate the functional $h(\theta_1,\phi_1)$ resulting from the above triplet for a general input state $\vert \psi \rangle = U_{\phi_0}U_{\theta_0}\vert 0\rangle$, considering some other pair MZI $(\theta_0,\phi_0)$ responsible for preparing this state.} Figure~\ref{fig: violations for any theta and phi} (c)-(f) shows plots corresponding to the violations of inequality \eqref{eq: c3 ineq 1} obtainable by different input states in the MZI. { In Fig.~\ref{fig: violations for any theta and phi} (c) we present the situation in which for \textit{all} values in the range $0 \leq \theta_1 \leq \pi$ and $0.042 \leq \phi_1 \leq \pi-0.084$, all choices of $U_{\theta_1}$ and $U_{\phi_1}$ incur in a violation of the inequality. }It is clear that there is an interesting trade-off between broad and high violations for states found. 

It is important to point that the investigation of contextuality in interferometers has been experimentally addressed before with the Kochen-Specker approach to contextuality~\cite{borges2014quantum,liu2009experimental}. See Ref.~\cite{Budroni21} for a recent comprehensive review with interferometric investigations of quantum contextuality; none of those reported in the review have used single photons in a simple MZI. The difficulty in probing this notion of contextuality on a single MZI stems from the fact that it is not applicable to two-dimensional systems, as already mentioned. 

\subsection{Advantages for information tasks over coherence-free and noncontextual models}\label{subsec: Generic advantages}

The simplicity of the argument for finding contextual violations using $C_3$ cycle graphs is built over strong structural results. However, the generality of the application, and the usefulness of such results may have a broader impact for generic proofs of contextual advantage provided by quantum theory, as it was already anticipated by Ref.~\cite{Lostaglio2020contextualadvantage}. The common nontrivial aspects of  proving contextual advantage of a given information task relies on finding the correct operational prepare-and-measure scenario for it, followed by later finding (possibly robust) noncontextuality inequalities, and relating those with a figure of merit that quantifies the success of the particular task~\cite{Lostaglio2020contextualadvantage,schmid2018discrimination,flatt2021contextual,carceller2022quantum,mukherjee2021discriminating}. This generic strategy is greatly simplified by the advent of a clear connection among contextuality scenarios, event graphs, and any figure of merit of quantum information tasks that can be described in terms of two-state overlaps. This precise connection allows to simplify this kind of result using the following approach:
\begin{enumerate}
    \item Given a particular quantum information task of interest, describe the figure of merit for this task in terms of two-state overlaps.
    \item Use the event graph inequalities to obtain bounds for the {  the specific figure of merit at study}. 
    \item Compare the optimal results obtainable using quantum theory with the inequality bounds. If there are quantum violations, there will be an advantage for the task, provided by either coherence or contextuality. 
\end{enumerate}
This strategy for contextual advantage was first envisioned in Ref.~\cite{Lostaglio2020contextualadvantage}, yet we extend it to the case of advantages provided by basis-independent coherence. As an example of this algorithmic approach, we apply this methodology for the MZI prepare-and-measure scenario to the task of quantum interrogation. We will see that, although   interrogation can be reproduced by a noncontextual model for some cases, as was shown in Ref.~\cite{catani2021interference}, there exists a quantifiable gap between efficiency achievable with such models and with quantum theory.

\subsubsection{Contextual advantage for quantum interrogation}\label{subsubsec: contextual advantage}

Consider the MZI from Fig.~\ref{fig:MZI}, the $?$-box being the photosensitive bomb associated with the quantum interrogation protocol \cite{elitzur1993quantum}. In this scenario, we set $U_{\theta_1} = U_{\theta_2}^\dagger$, with $U_{\theta_1}$ given by the general description of a BS in Eq.~\eqref{eq: BS}, in order to have a dark detector.

Let us now address operationally the content of the MZI statistics arising from quantum interrogation. The  efficiency of the task arises from a specific interplay between the creation of superposition in the preparation stage, and the ability to make a coherent detection in the measurement stage~\cite{masini2021coherence}. By using the formal operational terminology, we write $P_1$ to describe the preparation of a photon in mode \textit{a} which, in quantum theory, corresponds to the preparation of the state  $\vert 0 \rangle$. Similarly, with $P_2, P_3$ we refer to the preparations of a photon again in mode \textit{a}, followed by the two beam-splitters $BS_1$ and $BS_2$ as in Fig.~\ref{fig:MZI}, that in quantum theory are modelled as $U_{\theta_1} \vert 0 \rangle$ and $U_{\theta_2}\vert 0 \rangle$. Each preparation has associated a measurement procedure in the prepare-and-measure scenario we construct, exactly as we have described before when looking at the LSSS constraints and contextuality in the sequential setting. To ease notation with respect to the operational content of the scenario, we write $r_{ij} = p(0 \vert M_i, P_j)$, see table~\ref{tab:operational_content_table}. From the assumption $U_{\theta_1} = U_{\theta_2}^\dagger$, we can simply write $\theta_1 = \theta$ and $\theta_2 = \theta^\dagger$.

\begin{table}
    \centering
    \begin{tabular}{|c|c|c|c|}
    \hline
         & $P_1$ & $P_2$ & $P_3$ \\
    \hline
    $0|M_1$  &  $1$  &   $r_{\theta 0}$  &  $r_{\theta^\dagger 0}$   \\
    \hline
    $0|M_2$  &  $r_{\theta 0}$   &  $1$ &   $r_{\theta \theta^\dagger}$  \\
    \hline
    $0|M_3$  &  $r_{\theta^\dagger 0}$   &  $r_{\theta \theta^\dagger}$   &  $1$   \\
    \hline
    \end{tabular}
    \caption{\textbf{Operational symmetries over the statistics of quantum interrogation.} Notation for the prepare-and-measure statistics arising from the quantum interrogation scenario.}
    \label{tab:operational_content_table}
\end{table}

The statistics arising after the first BS corresponds to the statistics described by $r_{\theta 0} = p(0 \vert M_1,P_2), r_{\theta 1}=p(1\vert M_1,P_2)$. In quantum theory, { detecting the photon either in one arm or the other have the  associated probabilities} $r_{\theta 0}= \vert \langle \psi(\theta,0) \vert 0 \rangle \vert^2, r_{\theta 1} = \vert \langle \psi (\theta,0) \vert 1 \rangle \vert^2$. With reference to Fig.~\ref{fig:MZI}, in quantum theory $P_2$ corresponds to the preparation $\vert \psi(\theta,0)\rangle = U_\theta \vert 0 \rangle$, 
\begin{equation}\label{eq: preparation state bomb}
    \vert \psi(\theta,0)\rangle = \cos(\theta)\vert 0 \rangle + i\sin(\theta)\vert 1 \rangle. 
\end{equation}
In case of the presence of a bomb inside the device, the dark detector lights up when the photon is sent into the arm for which no explosion happens. In this event, we find the measurement statistics for the detectors to be $r_{\theta^\dagger 0} = p(0\vert M_3, P_1),r_{\theta^\dagger 1} = p(1\vert M_3, P_1) $ whereas, in quantum theory, $r_{\theta^\dagger 0} = \vert \langle 1 \vert \psi^\dagger(\theta,0)\rangle \vert^2$, $r_{\theta^\dagger 1} = \vert \langle 0 \vert \psi^\dagger(\theta,0)\rangle \vert^2$, with $\vert \psi^\dagger(\theta,0) \rangle = U_\theta^\dagger \vert 0 \rangle $. We define $P_3$ operationally just as $P_2$, but with respect to a BS that has a difference in phase of $\pi$. In quantum theory, this is described by $U_\theta^\dagger$, which gives
\begin{equation}\label{eq: detection state}
    \vert \psi^\dagger(\theta,0) \rangle = \cos(\theta)\vert 0 \rangle - i\sin(\theta)\vert 1 \rangle.
\end{equation}

Every preparation has an associated binary-outcome measurement.  Hence, $M_2$ and $M_3$ measurement procedures give $p(0 \vert M_2, P_2) = 1$ and $p(0 \vert M_3, P_3) = 1$, and the quantity that in quantum theory is given by $r_{\theta \theta^\dagger} = \vert \langle \psi (\theta,0) \vert \psi^\dagger(\theta,0) \rangle \vert^2$ is operationally defined by the confusability $p(0 \vert M_2,P_3) = r_{\theta \theta^\dagger}$.

The efficiency of the task thus addressed can be expressed solely by the operational quantities just described. The quantity $p_{succ} = r_{\theta 0}r_{\theta^\dagger 1}$ corresponds to the probability that the photon enters the MZI, does not chose the path with the bomb, and then, after the second BS, it makes the dark detector light up, ending up with a successful detection of the bomb without exploding it. On the contrary, $p_{bomb} = r_{\theta 1}$ corresponds to the probability that, after the first BS, the photon takes the path with the bomb, which consequently explodes, failing to accomplish the task. Hence,
\begin{equation}\label{eq: operational efficiency}
    \eta = \frac{r_{\theta 0}r_{\theta^\dagger 1}}{r_{\theta 0}r_{\theta^\dagger 1}+r_{\theta 1}}.
\end{equation}

It is interesting to note that there are operationally relevant symmetries that are respected by quantum theory, and that can be imposed in the prepare-and-measure scenario. In particular, quantum theory satisfies $r_{\theta 0} = \vert \langle \psi(\theta,0) \vert 0\rangle \vert^2 = \vert \langle \psi^\dagger(\theta,0) \vert 0 \rangle \vert^2 = r_{\theta^\dagger 0}$ and analogously,  $r_{\theta 1} = r_{\theta^\dagger 1}$. Moreover, $r_{\theta \theta^\dagger} = \vert \langle \psi(\theta,0)\vert \psi^\dagger(\theta,0) \rangle  \vert^2 = \vert\cos^2(\theta)-\sin^2(\theta)\vert^2= (r_{\theta 0} - r_{\theta 1})^2$. Therefore, we assume the following symmetries over the scenario, 
\begin{align}
    &r_{\theta 0} = r_{\theta^\dagger 0}\label{eq: operational constrain 1}\\
    &r_{\theta 1} = r_{\theta^\dagger 1}\label{eq: operational constrain 2}\\
    &r_{\theta \theta^\dagger} = (r_{\theta 0} - r_{\theta 1})^2\label{eq: operational constrain 3}
\end{align}
Under symmetries~\eqref{eq: operational constrain 1} and \eqref{eq: operational constrain 2} we can re-write the efficiency $\eta$ as a function of $r_{\theta 0}$ only,
\begin{equation}\label{eq: operational efficiency with constraints}
    \eta = \frac{r_{\theta 0}}{r_{\theta 0}+1},
\end{equation}
and with the symmetry \eqref{eq: operational constrain 3} and the $C_3$ inequality, it is possible to find a clear gap bounding the efficiency achievable by any noncontextual model that attempts at explaining the quantum theory predictions,
\begin{align*}
    -r_{\theta \theta^\dagger}+r_{\theta 0} + r_{\theta^\dagger 0} \leq 1 \hspace{0.8em}&\stackrel{\eqref{eq: operational constrain 1}}{\Rightarrow}\hspace{0.8em}
    -r_{\theta \theta^\dagger} + 2r_{\theta 0} \leq 1 \hspace{0.8em}\Rightarrow \\[1.5ex]
    \frac{r_{\theta 0}}{r_{\theta 0} + 1} \leq \frac{1 + r_{\theta \theta^\dagger}}{2(r_{\theta 0}+1)} \hspace{0.8em}&\stackrel{\eqref{eq: operational constrain 3}}{\Rightarrow}\hspace{0.8em} \eta \stackrel{\small{NC}}{\leq} \frac{1 + (2r_{\theta 0}-1)^2}{2(r_{\theta 0}+1)}.
\end{align*}

Figure~\ref{fig: efficiency standard interrogation} shows a plot of the efficiency $\eta$, achievable by quantum theory  as described by Eq.~\eqref{eq: operational efficiency with constraints}, versus the optimal noncontextual efficiency $\eta^{NC}_{opt} = \frac{1 + (2r_{\theta 0}-1)^2}{2(r_{\theta 0}+1)}$, as functions of the operational quantity $r_{\theta 0}$. From this plot, it is clear that the efficiencies meet at $1/3$, corresponding to the efficiency of the original proposal~\cite{elitzur1993quantum}, as well as that of the noncontextual model of Ref.~\cite{catani2021interference}. Our results parallel those of Ref.~\cite{catani2021interference}, which presents a noncontextual model for the cases when $r_{\theta 0}=1/2$ and $r_{\theta 0}=1$.

\begin{figure}[tb]
    \centering
    \includegraphics[width=\columnwidth]{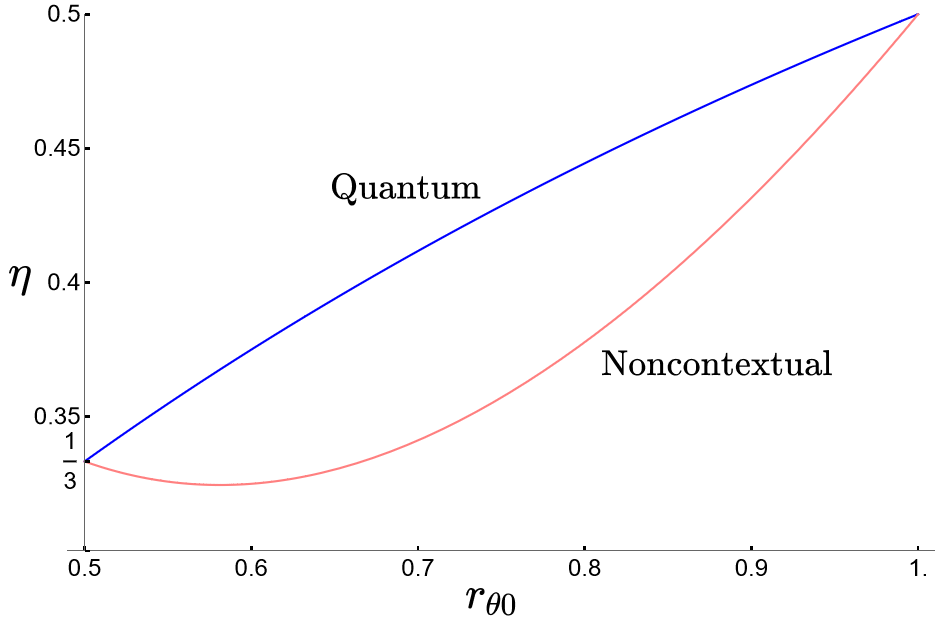}
    \caption{\textbf{Contextual advantage for quantum interrogation.} In blue, the upper curve presents the efficiency for an ideal action of asymmetric BSs in the MZI, achievable with standard quantum theory. In pink, the lower curve shows the optimal efficiency achievable by any noncontextual ontological model that explains the statistics of the associated prepare-and-measure scenario, while respecting two operational constraints also satisfied by quantum theory: $r_{\theta 0} = r_{\theta^\dagger 0}, r_{\theta \theta^\dagger} = (2r_{\theta 0}-1)^2$.}
    \label{fig: efficiency standard interrogation}
\end{figure}
We observe that asymmetric BSs provide contextual advantage for the task described with quantum theory, performing significantly better for some choices of $r_{\theta 0}$. The maximum gap is given by $\max_{r_{\theta 0}}(\eta - \eta_{opt}^{NC}) \approx 0.071$ at $r_{\theta 0} =\sqrt{3}-1 \approx 0.73$. It is also unclear if quantum theory, or some other toy noncontextual models, could respect the constraints imposed and follow the optimal noncontextual path depicted in Fig.~\ref{fig: efficiency standard interrogation}. Although the scenario just described achieves maximal efficiency for contextual advantage with $\eta_{max} = 1/2$, it is possible to develop robust schemes of quantum interrogation that are more efficient devices different from the MZI~\cite{kwiat1995interactionfree,kwiat1999high,rudolph2000better}. Those different scenarios allow, theoretically, for efficiencies as high as $\eta = 1$.

\section{Discussion and future work}

In this work we have provided experimentally robust tools to witness and quantify basis-independent coherence and contextuality inside Mach-Zehnder interferometers, from the toolbox of event graph scenarios and prepare-and-measure contextuality scenarios. We have also devised a general methodology for rigorously investigating quantum advantages related to interference, and showed a contextual advantage for the task of quantum interrogation as a case-study.

As a future step, it would be interesting to find complementarity relations between basis-independent coherence witnesses and path-information. To do so, one must develop a resource-theory for the event graph witnesses; a research path with independent merits. It is also relevant to study how the witnesses compare with standard quantification approaches in the cases for which the states $\rho$ have large dimensions. 

Finally, a relevant open problem is whether \textit{any} event graph inequality can be interpreted as a preparation noncontextuality inequality with respect to the construction of the scenarios we have considered here. 

\section{Acknowledgements}

We would like to thank Vinicius Pretti Rossi and Antonio Zelaquett Khoury for helpful comments on earlier versions of this work.

We acknowledge financial support from FCT -- Fundação para a Ciência e a Tecnologia (Portugal) via PhD Grant SFRH/BD/151199/2021 (RW), via PhD Grant SFRH/BD/151190/2021 (AC), and via CEECINST/00062/2018 (EFG). This work was supported by the ERC Advanced Grant QU-BOSS, GA no. 884676.

\bibliographystyle{quantum}

\onecolumn
\appendix

\section{Violations represent lower bounds on distance quantifiers}\label{appendix: rigorous violations analysis}

We may define the degree of basis-independent coherence witnessed by a given quantity $r$ in terms of the minimal distance with respect to the set $C_G$,
\begin{equation}
    \mathcal{D}(r) := \min_{r^* \in C_G} D(r,r^*)
\end{equation}
for $D(r,r^*)$ some abstract notion of distance. Even though there is no \textit{formal} resource theory yet built from the structure described in Ref.~\cite{wagner2022inequalities}, the quantity $\mathcal{D}(r)$ would be a natural candidate for \textit{any} future one; Ref.~\cite[section VI, B]{Chitambar19} discusses at length this standard construction within quantum resource theories and Ref.~\cite{Duarte18} gives some examples of the same construction for the resource theory of pre/post-processing for generalized contextuality. Letting $h: \mathbb{R}^{|E(G)|} \to \mathbb{R}$ represent the convex-linear functional of any given event graph inequality of $C_G$. The violation of $h$ can be described as $\vert h(r) - h(\tilde{r})\vert = \vert h(r-\tilde{r})\vert \leq \Vert h\Vert_\infty \Vert r-\tilde{r} \Vert_2 \leq \Vert h\Vert_\infty \Vert r-\tilde{r} \Vert_1$, for some $\tilde{r}$ that achieves the facet-defining bound $h(\tilde{r}) = s$ with $s$ given by the particular inequality $h(r)\stackrel{NC}{\leq}s$. Choosing $D(r,r^*) := \Vert r-r^*\Vert_1$ we get that
\begin{align*}
    &\mathcal{D}(r) = \min_{r^* \in C_G} \Vert r-r^*\Vert_1 \geq \min_{r^* \in C_G}\frac{\vert h(r) - h(r^*)\vert }{\Vert h \Vert_\infty} \\
    &= \frac{1}{\Vert h \Vert_\infty}\min_{r^* \in C_G}\vert h(r) - h(r^*)\vert = \left\{\begin{matrix}\frac{1}{\Vert h \Vert_\infty}\vert h(r)-s\vert & r \notin C_G \\ 0, & r \in C_G\end{matrix}\right..
\end{align*}
Any violation of the event graph inequalities may serve as a lower bound for $\mathcal{D}(r)$ up to the factor $1/\Vert h \Vert_\infty$. This allows us to use inequality violations to \textit{quantify} the coherence related to the set of states chosen. 

\section{General connection between event graphs and prepare-and-measure scenarios}\label{appendix: general}

In this appendix we show \textit{a mapping} between event graphs and prepare-and-measure contextuality scenarios. Recall that any state $\rho$ can also be understood as a POVM element of a binary outcome measurement procedure $M = \{\rho, \mathbb{1}-\rho\}$. This implies that any event graph defines  the preparation procedures to be considered as well as such possible measurements to be performed. Assuming that those states are close to being pure, it is possible to claim the existence of operations satisfying the operational constraints described below, in general, different than what was considered in the main text. The relationship can be described operationally as follows, for any given event graph $G$ with $G=(V,E)$:

\begin{enumerate}
    \item For each vertex $v \in V$ of $G$ we associate processes $P_v$ and measurement events $0\vert M_v$ of binary outcome measurements $M_v$. 
    \item For each edge $e \in E$ of $G$ we associate two other preparation processes $e= \{v,w\}$, $P_{v^\perp}^e, P_{w^\perp}^e$. 
    \item To each $4$-tuple of processes $(P_v,P_{v^\perp}^e,P_w,P_{w^\perp}^e)$ for each edge $e \in E$ we associate the operational equivalence $\frac{1}{2}P_v + \frac{1}{2}P_{v^\perp}^e \simeq \frac{1}{2}P_w + \frac{1}{2}P_{w^\perp}^e$.
    \item Each measurement procedure $M_v$ satisfy the operational equivalences that $p(0 \vert M_v,P_v) \geq 1-\varepsilon_v$ and $p(0 \vert M_v,P_{v^\perp}^e) \leq \varepsilon_v$ for every $e \in E$ for some real parameter $\varepsilon_v \geq 0$.
\end{enumerate}
In terms of the notation introduced in the main text we can write the above construction to define what we refer generically as Lostaglio-Senno-Schmid-Spekkens (LSSS) scenarios $\mathbb{B}_{LSSS} = (N_v + N_vN_e,N_v,2,\mathcal{OE}_P,\emptyset)$ where $\mathcal{OE}_P$ correspond to the operational equivalences just described. Assuming more operational equivalences in $\mathcal{OE}_P$ between the elements $P_{v^\perp}^e$ is also possible, as was done in the main text, which may simplify the scenario. We also assume that the scenarios $\mathbb{B}_{LSSS}$ satisfy the constraints of $p(0\vert M_v,P_v)\geq 1-\varepsilon_v$ and $p(0\vert M_v,P_{v^\perp}^e) \leq \varepsilon_v$ since this relation is somewhat needed for the structure of noncontextual models explaining the so-called operational confusabilities $r_{e} \equiv p(0\vert M_v,P_w)$ for every edge $e=\{v,w\}$. We have then the following result. We are using notation present in Refs.~\cite{chaturvedi2021characterising,wagner2021using} for the scenarios.

\begin{theorem}[Adapted from Ref.~\cite{Lostaglio2020contextualadvantage}]\label{theorem: lostagio_senno}
Let $\mathbb{B}_{LSSS}$ be the operational scenario described before. Then, a generalized noncontextual ontological model explaining the statistics of $p(0\vert M_v,P_w) \equiv p(M_v\vert P_w)$ must satisfy the following inequality,
\begin{equation}
    \vert \Vert \mu_v - \mu_w \Vert_1 - 2(1-p(M_v \vert P_w)) \vert \leq 2\varepsilon_v 
\end{equation}
in the limit of ideal measurements $\varepsilon_i \to 0$ we have,
\begin{equation}
    p(M_i \vert P_j) = 1 - \frac{1}{2}\int_\Lambda \vert \mu_i(\lambda) - \mu_j(\lambda)\vert \mathrm{d}\lambda.
\end{equation}
\end{theorem}

This kind of structural result for the confusabilities was first introduced in Ref.~\cite{schmid2018discrimination} using infinitely many operational equivalences. The above theorem directly implies, as a corollary, that the $n$-cycle event graph inequalities of $C_n$ provide us with robust generalized noncontextuality inequalities,

\begin{corollary}
{Let $G=C_n$ be a cycle graph with $n\geq 3$. Then, the $n$-cycle overlap inequalities, corresponding to the non-trivial facet-defining inequalities of the convex polytope $C_{C_n}$, can be mapped to robust noncontextuality inequalities of the scenario $\mathbb{B}_{LSSS}$,}
\begin{align*}
    r_{12}+r_{23}+\dots+r_{(n-1)n}-r_{n1} \leq n-2 + \varepsilon_1+\varepsilon_2+\varepsilon_3+\dots+\varepsilon_n
\end{align*}
and sign permutations of the r.h.s. 
\end{corollary}

\begin{proof}
Apply the triangle inequality satisfied by the $l_1$-norm, similarly to what has already been done in Ref.~[Appendix E]\cite{wagner2022inequalities}, while now allowing $\varepsilon_v \neq 0$ and therefore using the robust bounds provided by theorem~\ref{theorem: lostagio_senno}. Every confusability $r_{ii+1}$ satisfy,
\begin{equation*}
    -\varepsilon_i+1-\frac{1}{2} \Vert \mu_i - \mu_{i+1}  \Vert_1 \leq r_{ii+1}\leq \varepsilon_i + 1 - \frac{1}{2}\Vert \mu_i - \mu_{i+1}  \Vert_1.
\end{equation*}
Using this and applying the triangle inequality to $\Vert \mu_1 - \mu_n + (\mu_2-\mu_2+\dots+\mu_{n-1}-\mu_{n-1})\Vert_1$ we have the result.
\end{proof}

When we treat quantum theory as an operational theory, the confusabilities $p(M_i \vert P_j)$ match the two-state overlaps related to processes $P_i$ in the equivalence class defined by $\vert \phi_i \rangle$. 

The mapping from event graphs into prepare-and-measure scenarios is general. However, only the cycle-inequalities have been connected to noncontextuality inequalities. If \textit{any} event graph inequality can also be understood as a noncontextuality inequality for such scenarios, under the symmetries respected by confusabilities $r_{ij} = r_{ji}$ is an open problem.

\section{MZI with a generic state in input}

We can use the so-called double Mach-Zehnder interferometer~\cite{chrysosthemos_quantum_2022} to input general states in a standard MZI, see Fig.~\ref{fig:doubleMZI} for the construction. This is a simple way of constructing a MZI with general input states. 

\begin{figure}[H]
    \centering
    \includegraphics[width=0.7\textwidth]{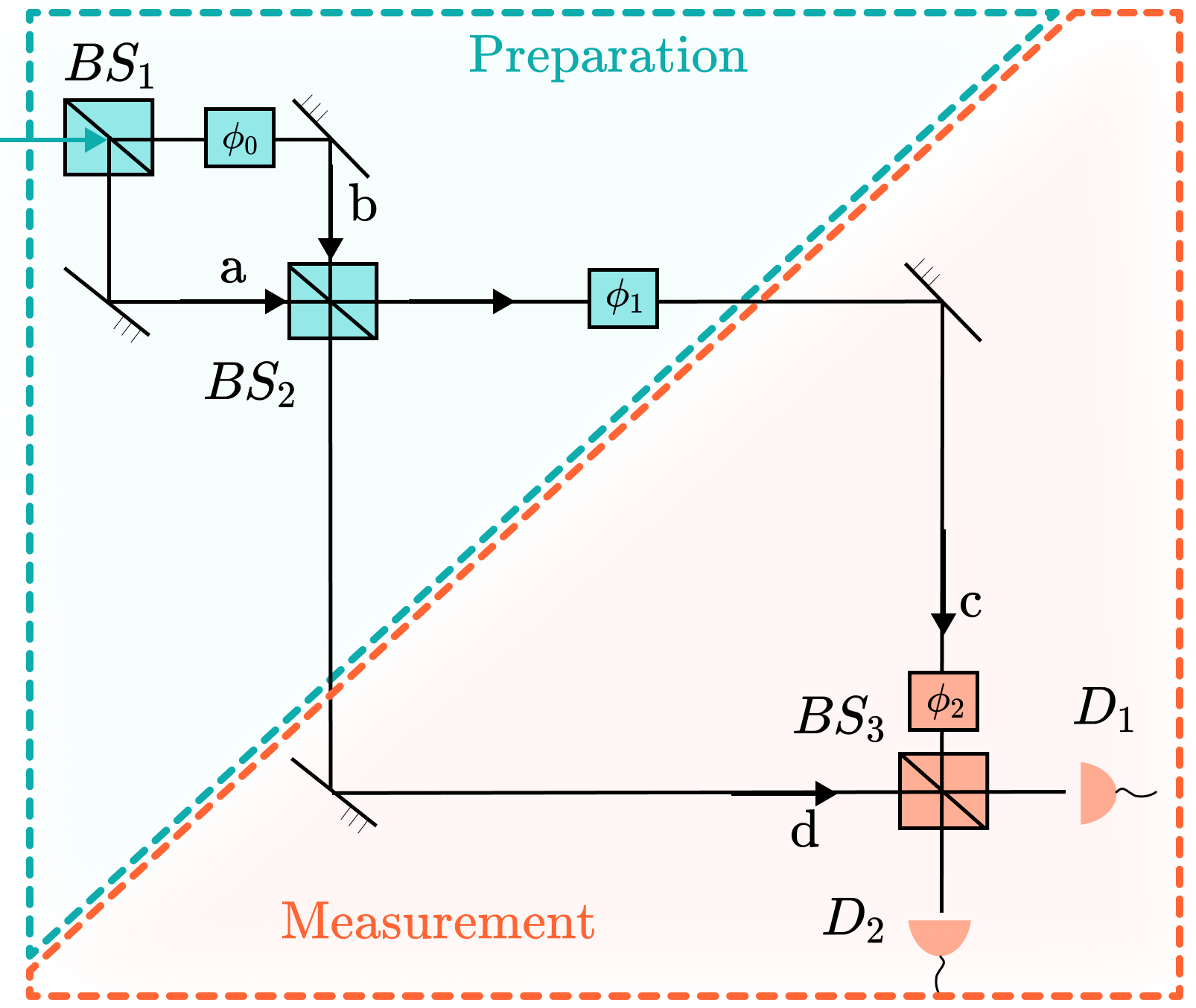}
    \caption{\textbf{Mach-Zehnder interferometer with generic input state.}}
    \label{fig:doubleMZI}
\end{figure}

When we refer to the MZI, we mean the two BSs and PSs $BS_1,BS_2$ and $\phi_1,\phi_2$ depicted in Fig.~\ref{fig:doubleMZI} while we simply consider the first stage as a universal single-qubit state preparation tool, described by $BS_1$ and $\phi_0$. In this way we may input a MZI with  $\ket{\psi}=\alpha\ket{0}+\beta\ket{1}$.

\noindent The first MZI creates a general state 
\begin{align}
    \ket{\psi_1}=  \cos\theta_0 \ket{0} + ie^{i\phi_0}\sin\theta_0\ket{1},
\end{align}
where $\ket{0}$ denotes the \textit{transmitted path}, meaning the upper path in the first MZI and the lower in the second, while $\ket{1}$ the complementary. After the second BS we have
\begin{align}
    \ket{\psi_2} &= \cos\theta_0(\cos\theta_1\ket{0}\nonumber+i\sin\theta_1\ket{1}) +\nonumber\\ &+ie^{i\phi_0}\sin\theta_0(\cos\theta_1\ket{1}\nonumber+i\sin\theta_1\ket{0}) = \nonumber\\
    &= (\cos\theta_0\cos\theta_1-e^{i\phi_0}\sin\theta_0\sin\theta_1)\ket{0}\nonumber\\
    &+i(\cos\theta_0\sin\theta_1+e^{i\phi_0}\sin\theta_0\cos\theta_1)\ket{1}
    \end{align}
Finally, the third state is given by the action of the PS described with an angle $\phi_1$
\begin{align}
    &\ket{\psi_3}=(\cos\theta_0\cos\theta_1-e^{i\phi_0}\sin\theta_0\sin\theta_1)\ket{0}\nonumber\\
    &+ie^{i\phi_1}(\cos\theta_0\sin\theta_1+e^{i\phi_0}\sin\theta_0\cos\theta_1)\ket{1}.
\end{align}

Summarizing, we have
\begin{align*}
    \ket{\psi_1}&= \cos\theta_0 \ket{0} + ie^{i\phi_0}\sin\theta_0\ket{1}\\
    \ket{\psi_2}&=(\cos\theta_0\cos\theta_1-e^{i\phi_0}\sin\theta_0\sin\theta_1)\ket{0}+i(\cos\theta_0\sin\theta_1+e^{i\phi_0}\sin\theta_0\cos\theta_1)\ket{1}\\
    \ket{\psi_3}&=(\cos\theta_0\cos\theta_1-e^{i\phi_0}\sin\theta_0\sin\theta_1)\ket{0}+ie^{i\phi_1}(\cos\theta_0\sin\theta_1+e^{i\phi_0}\sin\theta_0\cos\theta_1)\ket{1}
\end{align*}
Since we are interested in evaluating the inequality $r_{12}+r_{13}-r_{23}$, let us now compute the square modulus of the corresponding overlaps. A simple, but tedious, calculation gives
\begin{align*}
    |\langle \psi_1\vert \psi_2 \rangle|^2 = \underbrace{(\cos^2\theta_0 + \sin^2\theta_0)^2}_{=1} \cos^2\theta_1 + 2 \underbrace{\cos^2\theta_0 \sin^2\theta_0}_{\frac{1}{4}\sin^2(2\theta_0)} \sin^2\theta_1 - 2 \underbrace{\cos^2\theta_0 \sin^2\theta_0}_{\frac{1}{4}\sin^2(2\theta_0)} \sin^2\theta_1 \cos(2\phi_0)
\end{align*}
hence,
\begin{align*}
    |\langle \psi_1\vert \psi_2 \rangle|^2 =  \cos^2\theta_1 + \frac{1}{2} \sin^2(2\theta_0) \sin^2\theta_1 - \frac{1}{2} \sin^2(2\theta_0) \sin^2\theta_1 \cos(2\phi_0)=\\\cos^2\theta_1+\frac{1}{2} \sin^2(2\theta_0) \sin^2\theta_1(1-\cos(2\phi_0)),
\end{align*}
\begin{align*}
    |\langle\psi_1\vert\psi_3\rangle|^2 = &(\cos^4\theta_0 + \sin^4\theta_0)\cos^2\theta_1 + 2 \cos^2\theta_0 \sin^2\theta_0 \sin^2\theta_1 +\\
    &+2\cos\theta_0 \sin\theta_0 (\sin\theta_1 (-((\cos\theta_0 - \sin\theta_0)(\cos\theta_0 + \sin\theta_0)\cos\theta_1
    (\cos\phi_0 - \cos(\phi_0 - \phi_1))) -\\
    &-\cos\theta_0 \sin\theta_0 \sin\theta_1 \cos(2\phi_0 - \phi_1)) + \cos\theta_0 \sin\theta_0 \cos^2\theta_1 \cos\phi_1)
\end{align*}
\begin{align*}
     |\langle\psi_2 \vert \psi_3 \rangle |^2 =& (\cos^4\theta_0 + \sin^4\theta_0) \cos^4\theta_1 + 8 \cos^2\theta_0 \sin^2\theta_0 \cos^2\theta_1 \sin^2\theta_1 +(\cos^4\theta_0 + \sin^4\theta_0) \sin^4\theta_1 -\\
     & -4 \cos^2\theta_0 \sin^2\theta_0 \cos^2\theta_1 \sin^2\theta_1 \cos(2\phi_0) (-1 + \cos\phi_1) + 2 (\cos^2\theta_0 \sin^2\theta_0 \cos^4\theta_1 +\\
     &+ (\cos^2\theta_0 - \sin^2\theta_0)^2 \cos^2\theta_1 \sin^2\theta_1+ \cos^2\theta_0 \sin^2\theta_0 \sin^4\theta_1) \cos\phi_1 - 8 \cos\theta_0 (\cos\theta_0 - \sin\theta_0)\cdot\\
     &\cdot\sin\theta_0 (\cos\theta_0 + \sin\theta_0) \cos\theta_1 (\cos\theta_1 - \sin\theta_1) \sin\theta_1 (\cos\theta_1 + \sin\theta_1) \cos\phi_0 \sin^2\frac{\phi_1}{2},
\end{align*}
from which we obtain
\begin{align*}
    h_1 =& |\langle \psi_1\vert \psi_2\rangle |^2 +|\langle \psi_1\vert \psi_3\rangle |^2- |\langle \psi_2 \vert \psi_3 \rangle |^2 = \\
    =& \frac{1}{4}((7 + \cos(4\theta_0)) \cos^2\theta_1 - (3 +\cos(4\theta_0))\cos^4\theta_1 - (\cos\phi_1 (4 \cos^2(2\theta_0) + (1 + 3 \cos(4\theta_0)) \cos(2\theta_1)) +\\
    &+8 \cos^2(\theta_0) (\cos(2\phi_0) + \cos(2\phi_0 - \phi_1) +2\cos(2\theta_1))\sin^2\theta_0) \sin^2\theta_1 - (3 + \cos(4\theta_0)) \sin^4\theta_1 \\
    &+ (\cos(\phi_0 - \phi_1) +\cos\phi_0 (-1 + 4\cos(2 \theta_1) \sin^2\frac{\phi_1}{2})) \sin(4\theta_0) \sin(2\theta_1) +\\
    &+\cos(2\phi_0) (-1 + \cos\phi_1) \sin^2(2\theta_0) \sin^2(2\theta_1)),
\end{align*}
\begin{align*}
    h_2 =& h_3 = \mp |\langle \psi_1 \vert \psi_2 \rangle |^2 \pm |\langle \psi_1 \vert \psi_3 \rangle |^2+
    |\langle \psi_2 \vert \psi_3 \rangle |^2 = \\
    =&\frac{1}{4}(\cos^4\theta_0 (3 + \cos(4\theta_1)) + (3 + \cos(4\theta_1)) \sin^4\theta_0 + \cos\phi_1 (\cos(2\theta_1) \sin^2(2\theta_0) \\
    &+ 8\cos^2(2\theta_0) \cos^2\theta_1 \sin^2\theta_1) +\\
    &+2 \cos^2\theta_0 \sin^2\theta_0 [5 \cos\phi_1 - 4 \cos^2\theta_1 + \cos\phi_1 \cos(4\theta_1) +\\
    &+ 4 (-\cos(2 \phi_0 - \phi_1) +\cos(2 \phi_0) (1 - 2 (-1 + \cos\phi_1)\cdot \cos^2\theta_1)) \sin^2\theta_1] +\\
    &+(\cos(\phi_0 - \phi_1) - \cos\phi_0 (1 + 4 \cos(2\theta_1) \sin^2\frac{\phi_1}{2})) \sin(4\theta_0) \sin(2\theta_1) + 2\sin^2(2\theta_0)\sin^2(2\theta_1)).
\end{align*}

We have used $h_1$ in Fig.~\ref{fig: violations for any theta and phi} (c)-(f) using $\theta_0 = \pi/4$, $\phi_0 = \phi$ while $(\theta_1,\phi_1)$ remain the same as in the main text.

\end{document}